\def\llncs{0}
\def\fullpage{1}
\def\anonymous{0}
\def\authnote{0}
\def\notxfont{0}
\def\submission{0}

\ifnum\submission=1
\def\llncs{1}
\fi

\ifnum\llncs=1 	\documentclass[envcountsect,a4paper,runningheads,10pt]{llncs}
\else
	\documentclass[letterpaper,hmargin=1.0in,vmargin=1.0in,11pt]{article}
			\ifnum\fullpage=1
		\usepackage{fullpage}
		\fi
\fi

\usepackage[%
  colorlinks=true,
  citecolor=blue,
  pagebackref=true
]{hyperref}

\usepackage{amsmath, amsfonts, amssymb, mathtools,amscd, dsfont}

\usepackage{amsthm}

\usepackage{lmodern}
\usepackage[T1]{fontenc}
\usepackage[utf8]{inputenc}

\usepackage{arydshln} 
\usepackage{url}
\usepackage{ifthen}
\usepackage{bm}
\usepackage{multirow}
\usepackage[dvips]{graphicx}
\usepackage[usenames]{color}
\usepackage{xcolor,colortbl} 
\usepackage{threeparttable}
\usepackage{comment}
\usepackage{paralist,verbatim}
\usepackage{cases}
\usepackage{booktabs}
\usepackage{braket}
\usepackage{cancel} 
\usepackage{ascmac} 
\usepackage{framed}
\usepackage{authblk}
\usepackage{pifont}
\usepackage{qcircuit}
\usepackage{tikz}
\definecolor{darkblue}{rgb}{0,0,0.6}
\definecolor{darkgreen}{rgb}{0,0.5,0}
\definecolor{maroon}{rgb}{0.5,0.1,0.1}
\definecolor{dpurple}{rgb}{0.2,0,0.65}

\usepackage[capitalise,noabbrev]{cleveref}
\usepackage[absolute]{textpos}
\usepackage[final]{microtype}
\usepackage[absolute]{textpos}
\usepackage{everypage}
\DeclareMathAlphabet{\mathpzc}{OT1}{pzc}{m}{it}

\usepackage{algorithmic}
\usepackage{algorithm}
\usepackage{here}

\newtheoremstyle{thicktheorem}%
{\topsep}
{\topsep}
{\itshape}{}%
{\bfseries}%
{.}
{ }%
{\thmname{#1}\thmnumber{ #2}%
		\thmnote{ (#3)}%
}

\newtheoremstyle{remark}
{\topsep}
{\topsep}
	{}
	{}
	{}
	{.}
	{ }
	{\textit{\thmname{#1}}\thmnumber{ #2}
			\thmnote{ (#3)}%
	}

\ifnum\llncs=0
	\theoremstyle{thicktheorem}
	\newtheorem{theorem}{Theorem}[section]
	\newtheorem{lemma}[theorem]{Lemma}
	\newtheorem{corollary}[theorem]{Corollary}
	
	\newtheorem{definition}[theorem]{Definition}

	\theoremstyle{remark}
	
	\newtheorem{remark}[theorem]{Remark}

\else
\fi
\newtheorem{MyClaim}[theorem]{Claim}
\Crefname{MyClaim}{Claim}{Claims}

	\crefname{theorem}{Theorem}{Theorems}
	\crefname{assumption}{Assumption}{Assumptions}
	\crefname{construction}{Construction}{Constructions}
	\crefname{corollary}{Corollary}{Corollaries}
	\crefname{conjecture}{Conjecture}{Conjectures}
	\crefname{definition}{Definition}{Definitions}
	\crefname{exmaple}{Example}{Examples}
	\crefname{experiment}{Experiment}{Experiments}
	\crefname{counterexample}{Counterexample}{Counterexamples}
	\crefname{lemma}{Lemma}{Lemmata}
	\crefname{observation}{Observation}{Observations}
	\crefname{proposition}{Proposition}{Propositions}
	\crefname{remark}{Remark}{Remarks}
	\crefname{claim}{Claim}{Claims}
	\crefname{fact}{Fact}{Facts}
	\crefname{note}{Note}{Notes}

\ifnum\llncs=1
 \crefname{appendix}{App.}{Appendices}
 \crefname{section}{Sec.}{Sections}
\else
\fi

\ifnum\llncs=1
\renewcommand*{\backref}[1]{}
\else
	\renewcommand*{\backref}[1]{(Cited on page~#1.)}
	\ifnum\notxfont=1
	\else
		\usepackage{newtxtext}
	\fi
\fi
\ifnum\authnote=0  
\newcommand{\mor}[1]{}
\newcommand{\minki}[1]{}
\newcommand{\takashi}[1]{}
\newcommand{\dakshita}[1]{}
\newcommand{\kabir}[1]{}
\newcommand{\shira}[1]{}
\newcommand{\alper}[1]{}
\newcommand{\revise}[1]{}
\newcommand{\taiga}[1]{}
\else
\newcommand{\mor}[1]{$\ll$\textsf{\color{red} Tomoyuki: { #1}}$\gg$}
\newcommand{\dakshita}[1]{$\ll$\textsf{\color{magenta} Dakshita: { #1}}$\gg$}
\newcommand{\takashi}[1]{$\ll$\textsf{\color{orange} Takashi: { #1}}$\gg$}
\newcommand{\kabir}[1]{$\ll$\textsf{\color{darkgreen} K: {#1}}$\gg$}
\newcommand{\revise}[1]{{\color{purple}#1}} 
\newcommand{\taiga}[1]{$\ll$\textsf{\color{violet} Taiga: { #1}}$\gg$}
\newcommand{\shira}[1]{$\ll$\textsf{\color{red} Yuki: { #1}}$\gg$}
\newcommand{\alper}[1]{$\ll$\textsf{\color{teal} Alper: { #1}}$\gg$}
\fi

\newcommand{\eval}{\mathsf{Eval}}

\newcommand{\SD}{\mathsf{SD}} 

\newcommand{\puzz}{\mathsf{puzz}}
\newcommand{\ans}{\mathsf{ans}}

\newcommand{\Samp}{\algo{Samp}}

\newcommand{\cA}{\mathcal{A}}
\newcommand{\cB}{\mathcal{B}}
\newcommand{\cC}{\mathcal{C}}
\newcommand{\cD}{\mathcal{D}}
\newcommand{\cE}{\mathcal{E}}

\newcommand{\cG}{\mathcal{G}}
\newcommand{\cH}{\mathcal{H}}

\newcommand{\cM}{\mathcal{M}}

\newcommand{\cO}{\mathcal{O}}
\newcommand{\cP}{\mathcal{P}}
\newcommand{\cQ}{\mathcal{Q}}
\newcommand{\cR}{\mathcal{R}}

\def\makeuppercase#1{
\expandafter\newcommand\csname tl#1\endcsname{\widetilde{#1}}
}

\def\makelowercase#1{
\expandafter\newcommand\csname tl#1\endcsname{\widetilde{#1}}
}

\newcommand{\N}{\mathbb{N}}

\newcommand{\bbI}{\mathbb{I}}

\newcommand{\bbS}{\mathbb{S}}
\newcommand{\bbX}{\mathbb{X}}
\newcommand{\bbO}{\mathbb{O}}

\newcommand{\regC}{\mathbf{C}}

\newcommand{\regD}{\mathbf{D}}

\newcommand{\regB}{\mathbf{B}}
\newcommand{\regA}{\mathbf{A}}

\newcommand{\secp}{\lambda}

\newcommand{\advB}{\mathcal{B}}

\newcommand*{\pp}{\keys{pp}}

\newcommand*{\keys}[1]{\mathsf{#1}}

\newcommand*{\algo}[1]{\ensuremath{\mathsf{#1}}}

\newenvironment{boxfig}[2]{\begin{figure}[#1]\fbox{\begin{minipage}{0.97\linewidth}
                        \vspace{0.2em}
                        \makebox[0.025\linewidth]{}
                        \begin{minipage}{0.95\linewidth}
            {{
                        #2 }}
                        \end{minipage}
                        \vspace{0.2em}
                        \end{minipage}}}{\end{figure}}

\newcommand{\bit}{\{0,1\}}

\newcommand{\Setup}{\algo{Setup}}

\newcommand{\Gen}{\algo{Gen}}

\newcommand{\Ver}{\algo{Ver}}

\newcommand{\TD}{\algo{TD}}

\newcommand{\Eval}{\algo{Eval}}

\newcommand{\negl}{{\mathsf{negl}}}

\newcommand{\poly}{{\mathsf{poly}}}

\newcommand{\bin}{\{0,1\}}

\DeclareMathOperator*{\Exp}{\mathbb{E}}

\makeatletter
\DeclareRobustCommand
  \myvdots{\vbox{\baselineskip4\p@ \lineskiplimit\z@
    \hbox{.}\hbox{.}\hbox{.}}}
\makeatother

\newcommand{\reg}[1]{\mathbf{#1}}    
\newcommand{\adv}{\mathcal{A}}
\newcommand{\obf}{\mathsf{Obf}}      

\DeclareMathOperator*{\Prr}{Pr}
\newcommand{\CCol}{\mathsf{Col}}
\newcommand{\DCol}{\mathsf{DCol}}
\newcommand{\QCol}{\mathsf{QCol}}
\newcommand{\DQ}{\mathsf{DQ}}
\newcommand{\obfC}{\widetilde{C}}

\ifnum\llncs=0
\title{
On the Cryptographic Futility of 
Non-Collapsing Measurements}
\fi

\ifnum\llncs=1
\title{On the Cryptographic Futility of 
Non-Collapsing Measurements}
\fi

\ifnum\anonymous=1
\ifnum\llncs=1
\author{\empty}\institute{\empty}
\else
\author{}
\fi
\else
\ifnum\llncs=1
\author{
Alper Cakan\inst{1},
Dakshita Khurana\inst{2}, Tomoyuki Morimae\inst{3}, Yuki Shirakawa\inst{3},
Kabir Tomer\inst{4}, and Takashi Yamakawa\inst{5}
}
\institute{
 Yukawa Institute for Theoretical Physics, Kyoto University, Kyoto, Japan 
}
\institute{University of Illinois, Urbana-Champaign}
\else
\author{
Alper \c{C}akan\thanks{ Carnegie Mellon University, USA. \url{acakan@andrew.cmu.edu}}~~~~~~
 Dakshita Khurana\thanks{NTT Research, USA and University of Illinois Urbana-Champaign, USA. \url{dakshita@illinois.edu}}~~~~~~
 Tomoyuki Morimae\thanks{Yukawa Institute for Theoretical Physics, Kyoto University, Japan. \url{tomoyuki.morimae@yukawa.kyoto-u.ac.jp}, \url{yuki.shirakawa@yukawa.kyoto-u.ac.jp}}
\\ 
Yuki Shirakawa$^{\ddagger}$~~~~~~
Kabir Tomer\thanks{University of Illinois Urbana-Champaign, USA. \url{ktomer2@illinois.edu}}~~~~~~  Takashi Yamakawa\thanks{NTT Social Informatics Laboratories, Tokyo, Japan. \url{takashi.yamakawa@ntt.com}}
}

\fi
\fi

\date{}

\begin{document}

\maketitle

\begin{abstract}
We investigate quantum analogues of collision resistance and 
obtain separations between quantum ``one-way'' and ``collision-resistant'' primitives. 

\begin{itemize}
\item Our first result studies one-wayness versus collision-resistance defined over quantum circuits that output classical strings. 
We show that there is a classical oracle $\cO$ relative to which (sub-exponentially secure) indistinguishability obfuscation and one-way permutations exist even against adversaries that make quantum queries to a non-collapsing measurement oracle, $\cQ^{\cO}$.
Very roughly, $\cQ^{\cO}$ outputs the result of multiple non-collapsing measurements on the output of any quantum $\cO$-aided circuit.

This rules out fully black-box {\em quantum} constructions of $Y$ from $X$ for any $X \in \{$indistinguishability obfuscation and one-way permutations, public-key encryption, deniable encryption, oblivious transfer, non-interactive ZK, trapdoor permutations, quantum money$\}, Y \in \{$collision-resistant hash functions, hard problems in SZK, homomorphic encryption, distributional collision-resistant puzzles$\}$.

\item 
Our second result studies one-wayness versus collision-resistance defined over quantum states. Here, we show that relative to the same classical oracle $\cO$, (sub-exponentially secure) indistinguishability obfuscation and one-way permutations exist even against adversaries that make quantum queries to a {\em cloning unitary} $\mathsf{QCol}^\cO$. Very roughly, this latter oracle implements a well-defined, linear operation to clone a subset of the qubits output by any quantum $\cO$-aided circuit.

This rules out fully black-box 
constructions of quantum lightning from public-key quantum money.

\end{itemize}

\end{abstract}

\ifnum\llncs=0
\newpage
\tableofcontents 
\newpage
\fi

\section{Introduction}

The relationship between \emph{one-wayness} and \emph{collision resistance} is foundational in cryptography: while one-way functions are minimal for a wide range of cryptographic tasks~\cite{FOCS:ImpLub89}, collision-resistant hash functions are believed to be strictly stronger, and underpin advanced primitives such as ``hash-and-sign'' digital signatures~\cite{hashandsign}, merkle trees~\cite{hashandsign} and two-round statistically-hiding commitments~\cite{C:HalMic96}. 
There is strong evidence for this belief, in the form of classical black-box separations between (indistinguishability obfuscation (iO) and) one-way functions or one-way permutations (OWPs), and collision-resistant hash functions~\cite{EC:Simon98,FOCS:AshSeg15,TCC:BitDeg19}. These qualitative separations have guided the design and analysis of cryptosystems for decades. However, with the advent of quantum computing, many classical assumptions and separations require re-evaluation.\\


\noindent{\bf Quantum Computable Collision-Resistant Hash Functions.}  
Quantum adversaries may exploit uniquely quantum resources, such as superposition queries, potentially voiding classical separations or creating new ones. In particular, defining and separating quantum analogues of collision resistance and one-wayness are essential both for understanding the foundations of quantum cryptography and for building secure quantum protocols.
While~\cite{JC:HosYam24} showed that there is no quantum fully-black-box reduction from quantumly computable collision-resistant hash functions to OWPs or even trapdoor permutations, 
their proof is unwieldy and fails to separate collision-resistance from iO, leaving it open to prove such a separation.\\

\noindent{\bf Cloning Resistance.}
At the same time, quantum primitives such as quantum lightning exhibit collision-resistance-like properties, yet do not immediately yield collision-resistant hash functions. To appreciate this, it is helpful to start with the simpler notion of public-key quantum money~\cite{wiesner83,STOC:AarChr12}. Here, a bank generates quantum banknotes that anyone can verify but no adversary can counterfeit. More precisely, given a valid banknote, it should be infeasible to produce another with the same serial number. This is best understood as a form of targeted cloning resistance: the challenge state that must be cloned is the specific one generated by the bank. 

Quantum lightning~\cite{JC:Zhandry21} strengthens this idea. Instead of restricting minting to the bank, anyone can generate a valid state along with a serial number, yet it remains infeasible to produce two states with the same serial number. This property corresponds to general cloning resistance, where the adversary wins if they clone {\em any} of the banknotes in the support of the sampler.

Despite the surface-level similarity with collision-resistant hash functions, constructing quantum lightning has proven incredibly challenging, and the only known construction~\cite{SZ25} uses powerful generic assumptions such as indistinguishability obfuscation along with concrete algebraic assumptions such as learning with errors.
On the other hand, ``just'' (post-quantum) iO and one-way functions suffice to build quantum money.
Then, given the relative ease of building quantum money, it is natural to ask whether one can separate quantum lightning from quantum money--alternatively, rule out constructions of quantum lightning from ``just'' (post-quantum) iO and OWPs.


\subsection{Our Results}
The primary contribution of this work is to formulate a unified theory of one-wayness versus collision-resistance and cloning-resistance for quantum primitives. 
In doing so, we also explore connections with the complexity classes $\mathsf{PDQP}$ (Product-Dynamic Quantum Polynomial Time)~\cite{ITCS:ABFL16}  and $\mathsf{SampPDQP}$~\cite{MorShiYam_PDQP} which relax $\mathsf{BQP}$ to allow multiple, non-collapsing measurements during computation. We discuss these connections later in this section.\\

\noindent{\bf A Collision-Finding Oracle.}
For our first separation, we  introduce a random oracle $\mathcal{O}$ and show that one-wayness of $\mathcal{O}$ is preserved even given (quantum) queries to a quantum collision-finding oracle $\mathsf{Col}^{\cO}$.
For any oracle $\cO$, $\mathsf{Col}^{\cO}$ is a classical oracle
that, given a classical description of a quantum circuit $C^{\cO}$ with $\cO$-gates
such that $C^{\cO}\ket{0...0} = \sum_s \alpha_s \ket{s}\ket{\psi_s}$, generates the state
$\sum_s \alpha_s \ket{s}\ket{\psi_s}^{\otimes2}$, measures all qubits in the computational basis, and outputs the measurement result. (For the definition of $\mathsf{Col}$, see \cref{def:Col}.)  
$\mathsf{Col}$ can break 
distributional collision-resistant hashing (dCRH)~\cite{STOC:DubIsh06,EC:BHKY19} 
(and even its weaker quantum variant, distributional collision-resistant puzzles (dCRPuzzs)~\cite{MorShiYam_PDQP}).\footnote{For the definitions
of dCRH and dCRPuzzs, see \cref{def:dCRH} and \cref{def:dCRPuzzs}, respectively.}
This also implies that $\mathsf{Col}$ can decide $\mathsf{SZK}$~\cite{ITCS:ABFL16}, 
and thus break homomorphic encryption, as well as non-interactive computational private
information retrieval~\cite{C:BogLee13,TCC:LiuVai16}.

We show the following.
\begin{theorem}[Informal]
\label{thm:informalCol}
There exists a classical oracle $\mathcal{O}$ relative to which (sub-exponentially secure) indistinguishability obfuscation and (sub-exponentially secure) one-way permutations exist, even against adversaries equipped with quantum queries to $\mathsf{Col}^\mathcal{O}$.
\end{theorem}

Furthermore, we show that relative to these oracles, one can actually build iO for oracle-aided circuits, thereby capturing existing techniques that combine iO with {\em non-black-box} use of one-way functions/permutations to obtain public key encryption~\cite{FOCS:GGHRSW13,STOC:SahWat14}, oblivious transfer~\cite{FOCS:GGHRSW13,STOC:SahWat14}, quantum money~\cite{JC:Zhandry21}
and a variety of advanced primitives like deniable encryption~\cite{STOC:SahWat14}, non-interactive zero-knowledge proofs~\cite{STOC:SahWat14} and injective trapdoor functions~\cite{STOC:SahWat14}.
As such, we obtain the following corollary.

\begin{corollary}[Informal]
There does not exist a fully black-box construction of any $Y \in \{$collision-resistant hash functions, distributional collision-resistant puzzles$\}$ from 
$X \in \{$(indistinguishability obfuscation and one-way permutations), public-key encryption, deniable encryption, oblivious transfer, non-interactive ZK, trapdoor permutations, quantum money$\}$.
\end{corollary}

Note that Shmueli and Zhandry~\cite{SZ25} recently constructed one-shot signatures from (subexponentially secure) iO and learning with errors. 
One-shot signatures allow a signer holding a quantum signing key to sign a single message, only once, with respect to a classical verification key. Since one-shot signatures (black-box) imply dCRPuzzs~\cite{MorShiYam_PDQP}, the above corollary shows that their result cannot be improved to rely only on iO and OWPs (at least when using iO in a black-box way).\\

\noindent{\bf A Cloning Oracle.}
Next, we investigate what it means to find collisions on {\em quantum states} themselves. Our motivation is to better understand ``fully quantum'' collision resistant primitives such as quantum lightning~\cite{JC:Zhandry21}, which do not appear to imply classical or quantum computable collision resistant hashing, and do not seem to be broken by a 
collision-finding oracle.
Recall that quantum lightning allows the public generation of quantum states, called ``lightning bolts'', each associated with an efficiently verifiable serial number, and where it is infeasible for anyone (even the generator) to create two valid quantum states (lightning bolts) with the same serial number.


We introduce a cloning unitary oracle $\mathsf{QCol}$ that given the description of a quantum circuit, measures part of its output and clones the remaining part. More formally, $\mathsf{QCol}^{\cO}$ is a unitary that given a circuit $C$, swaps $\ket{\bot}$ (for some special symbol $\bot$) and $\sum_s \alpha_s \ket{s}\ket{\psi_s}^{\otimes2}$, where $C^\cO\ket{0...0} = \sum_s \alpha_s \ket{s}\ket{\psi_s}$.
(For the definition of $\mathsf{QCol}$, see \cref{def:QCol}.)
In more detail, letting $\cO$ be the same set of oracles as before that enable one-wayness (and indistinguishability obfuscation), we also provide access to $\mathsf{QCol}^{\cO}$.
The presence of $\mathsf{QCol}^{\cO}$ clearly breaks any quantum lightning scheme. 
However, we show that (sub-exponentially secure) OWPs and (sub-exponentially secure) iO continue to exist even against adversaries that make quantum queries--indeed, even inverse, conjugate, or transpose queries--to $\mathsf{QCol}^{\cO}$. 

\begin{theorem}[Informal]
There exists a classical oracle $\mathcal{O}$ relative to which (sub-exponentially secure) indistinguishability obfuscation and (sub-exponentially secure) one-way permutations exist, even against adversaries equipped with quantum queries to the unitary $\mathsf{QCol}^\mathcal{O}$ (and its conjugate, transpose and inverse).
\end{theorem}

While the first two separations used classical oracles, the above theorem relies on a unitary oracle, thereby capturing a smaller class of ``black-box'' techniques. However, by allowing access to the unitary along with its conjugate, transpose and inverse, we obtain the strongest form of unitary separation~\cite{Zhamodel}. We are unaware of meaningful techniques that evade such a separation but do not evade a black-box separation with a classical oracle.
This theorem also gives the following corollary:
\begin{corollary}
    There does not exist a fully black-box construction of quantum lightning from (subexponentially secure) $X$, where $X \in \{$(indistinguishability obfuscation and one-way functions), public-key encryption, deniable encryption, oblivious transfer, non-interactive ZK, trapdoor permutations, quantum money$\}$.
\end{corollary}

\noindent{\bf Connections with PDQP.}
Finally, we consider a generalization of our collision-finding oracle $\mathsf{Col}$, which makes the connection with the complexity classes $\mathsf{PDQP}$ and $\mathsf{SampPDQP}$ more explicit.
These classes are defined using an oracle $\cQ$, which is called a non-collapsing measurement oracle~\cite{ITCS:ABFL16}, that has the hypothetical ability to sometimes make multiple non-collapsing measurements on a quantum state.
In more detail, $\cQ$ is equipped not only with unitary operations but also with two types of measurement processes:
(i) Projective (or standard) measurements, which collapse the quantum state in accordance with the conventional laws of quantum mechanics; and
(ii) Non-collapsing measurements, which yield a sample from the corresponding output distribution, while leaving the quantum state unaffected and able to be measured again with an independent outcome each time.
(For a definition of the non-collapsing measurement oracle $\cQ$, see \cref{def:noncollapsingmeasurementoracle}.)
Then $\mathsf{PDQP}$ and $\mathsf{SampPDQP}$ are classes of decision (resp., sampling) problems that can be solved with
a polynomial-time classical deterministic algorithm that can make a single query to $\cQ$.

Also, note that $\cQ$ can be used to implement $\mathsf{Col}$, and strictly generalizes $\mathsf{Col}$ by allowing repeated unitaries to be applied interspersed with multiple non-collapsing measurements. It is natural to ask whether this more powerful oracle that decides $\mathsf{PDQP}$ can break one-wayness.
We show the following.
\begin{theorem}[Informal]
There exists a classical oracle $\mathcal{O}$ relative to which (sub-exponentially secure) indistinguishability obfuscation and (sub-exponentially secure) one-way permutations exist, even against adversaries equipped with quantum queries to $\mathsf{\cQ}^\mathcal{O}$.
\end{theorem}

Over the past decade, iO has served as a useful tool to build hard problems in various complexity classes, such as $\mathsf{PPAD}$~\cite{FOCS:BitPanRos15}.
This theorem shows that black-box use of OWPs and iO (even for circuits that can implement the permutation) cannot be used to obtain hard problems in $\mathsf{PDQP}$.
Furthermore since queries to $\mathsf{\cQ}^{\mathcal{O}}$ can break $\mathsf{SZK}$, homomorphic encryption and non-interactive computational private
information retrieval, we have the following corollary.

\begin{corollary}[Informal]
There does not exist a fully black-box construction of hard problems in $\mathsf{SZK}$, homomorphic encryption and non-interactive computational private
information retrieval from (subexponentially secure) $X$, where $X \in \{$(indistinguishability obfuscation and one-way functions), public-key encryption, deniable encryption, oblivious transfer, non-interactive ZK, trapdoor permutations, quantum money$\}$.
\end{corollary}

\section{Technical Overview} 

\noindent{\bf Background on Black-Box Separations.}
A fully black-box construction of a primitive $\cP$ from another primitive $\cP'$ consists of two algorithms: a construction $C$ and a reduction $R$, each of which are ``black-box''. 
Namely, for any oracle $P'$ that enables a correct implementation of the primitive $\cP'$, $C^{P'}$ instantiates a construction of $\cP$, i.e., $C^{P'}$ satisfies the correctness of $\cP$. 
Additionally, 
for any oracle $A$ that breaks $\cP$-security of $C^{P'}$, the (efficient) reduction $R^{A,P'}$ breaks the $\cP'$-security of $P'$. 

The standard framework (e.g.in~\cite{EC:Simon98,TCC:BitDeg19}) to rule out fully black-box constructions is to construct oracles $(S,\Gamma)$ such that the following hold.
\begin{itemize}
    \item \textbf{Property 1:} There exists a construction $C_1^{S}$ that satisfies the correctness of $\cP'$, and is $\cP'$-secure against efficient adversaries that may query $S$ and $\Gamma$.
    \item \textbf{Property 2:} For any construction $C_2^{S}$ that satisfies correctness of $\cP$, there exists an efficient adversary $A^{S,\Gamma}$ that breaks the $\cP$-security of $C_2^{S}$.
\end{itemize}
Suppose there exists a fully black-box construction $(C,R)$ of $\cP$ from $\cP'$. Since $C_1^{S}$ satisfies $\cP'$-correctness, $C_2^{S} := C^{(C_1^{S})}$ satisfies $\cP$-correctness. Therefore by Property 2, there exists an adversary $A^{S,\Gamma}$ that breaks the $\cP$-security of $C_2^{S}$. By the definition of black-box reductions, $R^{A^{S,\Gamma},C_1^{S}}$ must break the $\cP'$-security of $C_1^{S}$. However, this gives an efficient adversary that queries $S$ and $\Gamma$ and breaks the $\cP'$-security of $C_1^{S}$, contradicting the first property above.

This paper will follow a similar methodology. We will show that a variety of primitives cannot be constructed from iO and OWPs (where the iO supports oracle-aided circuits, to capture existing techniques that combine iO with non-black-box use of one-way functions/permutations). The oracle $S$ is formally defined in Section \ref{sec:io}, while the oracle $\Gamma$ will depend on the primitive being separated, and will be defined for each primitive in Sections \ref{sec:CR}, \ref{sec:light}, and \ref{sec:PDQP}.

Before we explain what the oracles $(S, \Gamma)$ will look like for our separations, it is useful to recap the classical separation between one-wayness and collision resistance~\cite{EC:Simon98}. In what follows, we present a simplified informal exposition of Simon's classical separation, which to our knowledge, has not appeared in prior work.
Then, we describe the issues that arise when quantum queries are introduced.
\\

\noindent{\bf The Classical Setting.}
The overall idea when classically separating one-way functions/permutations and collision resistant hash functions, is to introduce a ``collision-finding'' oracle $\mathsf{Col}$ that on input a circuit $C$ will find and output {\em distributional} collisions in the circuit. 
Namely, the oracle will evaluate $C(\cdot)$ on a uniformly chosen input $x$ (that is fixed as a function of $C$), and then output random $x'$ such that $C(x) = C(x')$.

Let us now describe the separating oracles in the classical setting.
Since we eventually want to argue that one-way functions/permutations exist, we set $S= \cO$, a random oracle from $\{0,1\}^{n} \rightarrow \{0,1\}^{3n}$. We will set $\Gamma$ to be $\mathsf{Col}^{\cO}$. We will also allow queries to $\mathsf{Col}$ to be circuits $C^{\cO}$ that have $\cO$-oracle gates in them. As before, on input a circuit $C$ with polynomially many $\cO$-gates in it, $\mathsf{Col}^{\cO}$ will sample uniform input $x$, and output a pair $(x, x')$ such that $C^{\cO}(x) = C^{\cO}(x')$ (and where $x'$ is essentially a random preimage of $C^{\cO}(x)$).

The ``one-way'' security game involves picking a random $x$, and then challenging the adversary to invert $y = \cO(x)$, with only polynomially many queries to these oracles.
One way to argue that inversion is hard is to consider an experiment where the oracle $\cO$ is replaced with a different oracle $\cO'$. 
The oracle $\cO'$ is identical to $\cO$, except that for any $z$ such that $\cO(z) = y$, $\cO'(z)$ outputs a uniformly random value $y'$ (note that w.h.p. due to $\cO$ being expanding, there is only a single such $z$). 
Then, one-wayness can be proved by contradiction by combining the following arguments.
\begin{enumerate}
\item Any adversary that inverts $y$, given access to $\cO$ can be used to distinguish between $\cO$ and $\cO'$ (by checking the image of the adversary's output). 
\item The only way distinguish between $\cO$ and $\cO'$ is to actually query $\cO'$ on $z$ s.t. $\cO(z) = y$.
\item Since $\cO'$ is independent of $y$, the adversary cannot possibly query $\cO'$ on $z$ s.t. $\cO(z) = y$ (except with negligible probability).
\end{enumerate}

When the adversary can additionally access $\mathsf{Col}$, the second step of the argument above is not necessarily immediate.
Fortunately, we can show that a variant of the second step can be made to work by relying the following lemma (which is loosely inspired by a lemma developed in the context of separating 
$\mathsf{PDQP}$ from $\mathsf{NP}$~\cite{ITCS:ABFL16})\footnote{Note that abstracting out this lemma also helps obtain a simplified understanding of the classical separation~\cite{EC:Simon98} of collision-resistant hash functions from one-way functions. Also, we do not formally prove this lemma in this work, instead we prove a stronger quantum generalization--please refer to (Informal) Lemma \ref{lem:overview-2}.}.
\begin{lemma}[Informal]
\label{lem:overview-1}
For every pair of oracles $\cO$ and $\cO'$, for every oracle-aided circuit $C$ and every input $x$, the statistical distance between $\mathsf{Col}^{\cO}(C)$ and $\mathsf{Col}^{\cO'}(C)$ given the truth table of $\cO$ is at most thrice the statistical distance between $C^\cO(x)$ and $C^{\cO'}(x)$ for uniform $x$ given the truth table of $\cO$. 
\end{lemma}

This says that, if given a single query to the collision-finder, an adversary can find a point on which $\cO$ and $\cO'$ differ, then the circuit $C$ (provided as input to $\mathsf{Col}$) must already be distinguishing between $\cO$ and $\cO'$, and therefore already querying $\cO'$ on a point where $\cO$ and $\cO'$ differ. 
This fact can then be used to complete a proof similar to the one sketched above.

A similar analysis extends--via a standard hybrid argument--even when the distinguisher is allowed multiple queries to the collision-finder. It is possible to come up with a sequence of hybrids where the first $i$ oracle queries of the collision-finder are answered via $\cO$, and the following queries are simulated using $\cO'$ (while ensuring that any duplicate queries are answered consistently). Any distinguisher that eventually detects a switch from $\cO$ to $\cO'$ must also do so in some adjacent set of hybrids, which then allows us to argue about the existence of a circuit that finds a differing input between $\cO$ and $\cO'$ (as before).\\

\noindent{\bf Quantum Queries: Challenges and Solutions.}
First, let us understand how one-way functions can be obtained from random oracles with security under superposition queries (in the absence of any collision finder). The one-way function $f(x)$ is simply defined as $\cO(x)$.
Given a challenge $y = \cO(x)$ for a random $x$, we can create a punctured version of $\cO$  that outputs independent random values on $z$ such that $\cO(z) = y$ (again w.h.p. there is a single such $z$), and call this oracle $\cO'$. 
As before, an inverter can be used to derive a contradiction by combining the following three observations.
\begin{enumerate}
    \item Any adversary that inverts $y$ can be easily used to distinguish between $\cO$ and $\cO'$ (by checking the image of the adversary's output). 
    \item The quantum one-way-to-hiding lemma~\cite{EC:Unruh12} says that the only way to distinguish between $\cO$ and $\cO'$ is to query $\cO'$ (with noticeable query mass) on some $z$ such that $O(z) = y$. 
    \item Since $\cO'$ is independent of $y$, the adversary cannot possibly query $\cO'$ (with noticeable query mass) on $z$ s.t. $\cO(z) = y$ (except with negligible probability).
\end{enumerate}

Of course, our goal is to rule out {\em quantum} constructions of collision-resistant hash functions from one-way functions/permutations. 
To do this, we must also provide a collision-finding oracle to the adversary. To rule out quantum constructions, the collision-finding oracle $\mathsf{Col}$ from before must be updated to find collisions in {\em quantum} circuits that make superposition queries to $\cO$. That is, without loss of generality, an input to $\mathsf{Col}$ will take the form of a quantum circuit $C$ with $\cO$-gates, and denoting $C^{\cO}\ket{0...0} = \sum_s \alpha_s \ket{s}\ket{\psi_s}$, $\mathsf{Col}$ will output the result of measuring all qubits in the state $\sum_s \alpha_s \ket{s}\ket{\psi_s}^{\otimes 2}$ in the computational basis.
Furthermore to rule out quantum reductions, we will have to argue that quantum queries to $\mathsf{Col}$ will not help break one-wayness.

Unfortunately, it is unclear why the one-way-to-hiding property above would hold in the presence of quantum queries to $\mathsf{Col}$. 
The very first quantum query to $\mathsf{Col}$ could already distribute a small fraction of amplitude across every point in its domain. Furthermore, since $\mathsf{Col}$ is itself inefficient, it is not at all clear why its response to a superposition query would avoid leaking information about whether the oracle was $\cO$ or $\cO'$.





To address, this, we take inspiration from the compressed oracle technique~\cite{C:Zhandry19} and {\em purify} $\mathsf{Col}$ to obtain a superposition over truth tables, where every input $C$ is mapped to an output pure state in this ``purified'' oracle. Furthermore, queries to $\mathsf{Col}$ can be simulated lazily given the ability to compute/uncompute the corresponding pure state.
This ability is captured by a ``compression'' unitary, denoted by $W_{\cO}$, that swaps $\ket{\bot}$ to the corresponding pure state.
Our next step is to prove the following lemma, which is the quantum analogue of Lemma \ref{lem:overview-1}, and which says that a variant of one-way-to-hiding continues to hold in the presence of quantum queries to (the compression unitary of) $\mathsf{Col}$.
\begin{lemma}[Informal]
\label{lem:overview-2}
For every pair of oracles $\cO$ and $\cO'$, and every quantum query $\ket{\psi}$, if the trace distance between 
$W_{\cO}\ket{\psi}$
and
$W_{\cO'}\ket{\psi}$
(given $\cO$) is noticeable,
then $\ket{\psi}$
 has noticeable mass on a circuit $C$ for which the statistical distance between  $C^\cO\ket{0...0}$ and $C^{\cO'}\ket{0...0}$ is noticeable, given $\cO$. 
\end{lemma}
Intuitively, this lemma follows by a hybrid argument where the joint distribution of the outputs of $\mathsf{Col}$ are modified one part at a time (we refer the reader to the proof of \cref{clm:abcd} for a complete description).




The lemma above says that the adversary helps find a quantum circuit $C$ that without ever querying the compression unitary $W_{\cO}$, will distinguish between $\cO$ and $\cO'$. This in turn means that measuring a random query of $C$ (to $\cO'$) will, with noticeable probability, yield $z$ such that $\cO(z) = y$.
The above lemma suffices to show that a {\em single} quantum query to the collision-finder, even on quantum circuits, does not break one-wayness.
Due to the way these compression unitaries are defined, we are then able to switch $W^{\cO}\ket{\psi}$ to $W^{\cO'}\ket{\psi}$ query-by-query, and rely on a quantum hybrid argument~\cite{BBBV} to argue that even sub-exponentially many queries to $\mathsf{Col}$ do not help break one-wayness.  

We are also able to extend this to allow multiple non-collapsing measurements in between computation. That is, our oracle can be strengthened from $\mathsf{Col}$ to a non-collapsing-measurement oracle, that applies multiple unitaries interspersed with non-collapsing (and collapsing) measurements.
We do this by defining a class of unitary oracles which we refer to as one-way-to-hiding (OW2H) compatible, in \cref{sec:io}, and showing that one-way-to-hiding continues to hold in the presence of such unitaries. Subsequently we prove that the compression unitary for $\mathsf{Col}$, as well as for the more general non-collapsing measurement oracle, is OW2H-compatible.
\\

\noindent{\bf Enabling Permutations.}
Recall that our goal is to not only prove that one-way functions exist, but also that OWPs and iO exist in the presence of queries to $\mathsf{Col}$, $\mathsf{QCol}$, or a non-collapsing measurement oracle.

The argument for the existence of OWPs follows a similar template as one-way functions, except that the oracle $\cO$ is a random permutation oracle. To argue that $y = \cO(x)$ is hard to invert, we will again replace $\cO$ with a different oracle $\cO'$. The oracle $\cO'$ is identical to $\cO$, except that on a random $x'$, $O(x')$ outputs $y$. Then, one-wayness of the permutation can be proved by contradiction by combining the following arguments.
\begin{enumerate}
\item Any adversary that inverts $y$ given $\cO$ will without loss of generality query $\cO$ on $x$.
\item One-way-to-hiding: The only way to distinguish between $\cO$ and $\cO'$ is to query $\cO$ (with noticeable query mass) on $x'$.
\item Since $\cO$ is independent of $x'$, the adversary cannot possibly query $\cO$ on $x'$ (except with negligible probability).
\item By 2 and 3, the adversary cannot query $\cO'$ on $x'$ (except with negligible probability). But by 1, 2 and 3, the adversary must query $\cO'$ on $x$ (with noticeable probability). However, to an adversary that makes oracle calls to $\cO'$, $x$ and $x'$ appear symmetric: thus an adversary cannot query one with higher probability than the other.
\end{enumerate}
As before, the second step of this argument is not immediate when given additional oracles (like a collision-finding oracle) that depend on $\cO$ or $\cO'$. Fortunately, a OW2H-compatible unitary allows us to recover this step just like the case of one-way functions.\\

\noindent{\bf Enabling Obfuscation.}
To build obfuscation, following~\cite{FOCS:AshSeg15} and \cite{TCC:BitDeg19} we introduce an expanding random oracle $\mathsf{Obf}$ that maps input circuit-randomness pairs $(\mathfrak{C},r)$ to $\mathfrak{C}'$, which is interpreted as an obfuscation of $\mathfrak{C}$. In addition, an $\mathsf{Eval}$ oracle given $\mathfrak{C}'$ and input $x$ inverts $\mathsf{Obf}$ on $\mathfrak{C}'$ to recover $\mathfrak{C}$, then outputs $\mathfrak{C}(x)$.

Recall that to prove obfuscation security, for every pair of functionally equivalent circuits $\mathfrak{C}_0$ and $\mathfrak{C}_1$, we would like to show that an adversary given $\mathsf{Obf}(\mathfrak{C}_b;r)$ for uniform $b, r$ cannot predict $b$, even with access to $\mathsf{Eval}$ and the collision-finding oracle $\mathsf{Col}$.
Furthermore, any circuits $C$ on which $\mathsf{Col}$ is queried can now have $\mathsf{Obf}$ and $\mathsf{Eval}$ oracle gates as well.

In the classical literature this is accomplished by a series of hybrids involving reprogramming the oracles adaptively. Intuitively, we move to a world where the adversary is given a random string $z$ instead of the obfuscated circuit. The $\Eval$ oracle is then modified such that querying the oracle on $(z,x)$ returns $\mathfrak{C}_0(x)$.  Replicating a similar argument appears to be considerably more challenging in the quantum setting.

Thus, we take a different approach. We prove the security of the obfuscation scheme without requiring adaptive reprogramming in any of our hybrids. This yields a new proof idea that is arguably simpler than the existing classical analysis (eg., in~\cite{TCC:BitDeg19}).

Taking a step back, the adversary's goal in the iO security game is to distinguish $\obf(\mathfrak{C}_0; r^*)$ from $\obf(\mathfrak{C}_1; r^*)$ where $r^*$ is the randomness used by the challenger. If the adversary only had access to $\obf$ and not $\Eval$, this would be easy to show by the same techniques we use to show one-wayness of a random oracle. The presence of $\Eval$ complicates this approach, since $\Eval$ uses $\obf^{-1}$ to ``unobfuscate'' circuits in order to evaluate them.

Our goal will be to (indistinguishably) move to a world where $\mathsf{Obf}$ is replaced by $\mathsf{Obf}'$ which is identical to $\mathsf{Obf}$ except that $\mathsf{Obf}'(\widehat{\mathfrak{C}};r^*) = \bot$ for all strings $\widehat{\mathfrak{C}}$, where $r^*$ refers to the randomness used by the challenge obfuscation. Suppose we could accomplish this while allowing $\Eval$ to still use $\obf^{-1}$ to invert its inputs. 
Then for any two functionally equivalent circuits $\mathfrak{C}_0$ and $\mathfrak{C}_1$, strings of the form $\obf(\mathfrak{C}_b; r^*)$ are \textit{identically} distributed in the adversary's view, independent of the number of queries made by the adversary (and hence independent of $\mathsf{Col}$ which can be simulated using unbounded queries). This means the view of the adversary in the security game is independent of the challenger's choice bit, implying that the adversary has zero advantage in the distinguishing game. This is because $\obf'$ queries cannot help distinguish the obfuscations, while $\Eval$ behaves identically for obfuscations of functionally equivalent circuits. 

The argument above gets more complicated when we introduce a collision-finder. Recall that we have already discussed why the compression unitary for $\mathsf{Col}$ is OW2H-compatible; all that remains is to prove that the security of obfuscation is retained in the presence of any OW2H-compatible unitary. 
In particular, switching from $\obf$ to $\obf'$ in the presence of such a unitary is non-trivial.
Nevertheless, we show that this is possible 
by relying on the fact that the unitary is OW2H-compatible: namely 
distinguishing $\mathsf{Col}$ that uses $\mathsf{Obf}$ from $\mathsf{Col}$ that uses $\mathsf{Obf}'$ leads to an efficient algorithm that finds points where the two oracles differ. 



\paragraph{Breaking Quantum Lightning.}

Quantum lightning is a form of publicly verifiable quantum money which allows for the public generation of unclonable states. Namely, anyone can produce quantum bank-notes/``lightning'' states along with an efficiently verifiable classical serial number, such that it is infeasible for anyone (even the generator) to create two states that verify against the same serial number.
Quantum lightning is an instance of collision/cloning resistance over quantum states, as opposed to classical strings. Our final technical goal is to demonstrate that this type of cloning resistance {\em also} cannot be based on iO and OWPs.

To achieve this, we will first construct a quantum oracle capable of breaking any quantum lightning scheme. Imagine an oracle that on input quantum circuits representing a  minting algorithm, ran the minting algorithm to obtain a serial number $s$ and a banknote $\ket{\psi_s}$, and then returned $s$ along with \textit{two} copies of $\ket{\psi_s}$. Such an oracle would break any lightning scheme. However, it does not seem powerful enough to break any quantum money scheme, since in quantum money one needs to clone banknotes corresponding to an externally sampled serial number $s$.
This difference is crucial since publicly verifiable quantum money \textit{can} be built from iO and OWPs~\cite{JC:Zhandry21}. 

Can we just use this oracle? The first issue with defining an oracle this way is that the oracle is not unitary. We can, however, consider a purified version of this oracle that outputs a pure state $\ket{\phi}:= \sum_s \alpha_s \ket{s}\ket{\psi_s}\ket{\psi_s}$, where the output of the (purified) minting algorithm was $\sum_s \alpha_s \ket{s}\ket{\psi_s}$.
Now, measuring the first register of $\ket{\phi}$ yields the same marginal distribution over serial numbers as the mint algorithm. 

In order for this oracle to be unitary, we define it as the operation that swaps between states $\ket{\bot}$ and $\ket{\phi}$. Note here the serendipitous similarity to the compression unitary for $\mathsf{Col}$, which similarly swaps $\ket{\bot}$ to the purified state of $\mathsf{Col}$! In fact, we can show that this lightning cloning unitary also satisfies the ``one-way to hiding'' (OW2H) compatibility described before. This, combined with the fact that iO and OWPs continue to remain secure even given queries to a OW2H unitary, helps complete our proof.

\paragraph{Follow-up Work.}
An upcoming work~\cite{personalcomm} builds on our technical toolkit, using it to separate constant-round proofs of quantumness from 
iO and OWPs. They crucially rely on the fact that our arguments do not require adaptive reprogramming.

\paragraph{Roadmap.} 
The rest of this paper is arranged as follows. In Section \ref{sec:preliminaries} we define notation, known primitives and lemmas that will be useful later. 
In Section \ref{sec:qq}, we introduce some helpful techniques to analyse algorithms that make quantum queries to oracles. 

Section \ref{sec:io} defines the class of  OW2H-compatible unitaries and shows how to realize a construction of iO and OWPs that {\em remains secure} in the presence of a class of unitary oracles which we refer to as one-way-to-hiding compatible.  

The remaining sections rule out black-box constructions of various primitives from iO and OWPs by constructing oracles that break the security of these primitives. The oracles are either OW2H-compatible or can be simulated by OW2H-compatible oracles, and therefore do not break the security of iO and OWPs.

Section \ref{sec:CR} contains the separations for collision-resistant primitives, Section \ref{sec:light} contains the separations for quantum lightning, and finally in Section \ref{sec:PDQP} we show that iO and OWPs remain secure even against (quantum) queries to a non-collapsing-measurement oracle.\\

\section{Preliminaries}\label{sec:preliminaries}

\subsection{Basic Notations} 
We use standard notations of quantum computing and cryptography.
For a bit string $x$, $|x|$ is its length.
$\mathbb{N}$ is the set of natural numbers.
We use $\secp$ as the security parameter.
$[n]$ means the set $\{1,2,...,n\}$.
For a finite set $S$, $x\gets S$ means that an element $x$ is sampled uniformly at random from the set $S$.
$\negl$ is a negligible function, and $\poly$ is a polynomial.
All polynomials appear in this paper are positive, but for simplicity we do not explicitly mention it.
PPT stands for (classical) probabilistic polynomial-time and QPT stands for quantum polynomial-time. 
For an algorithm $\cA$, $y\gets \cA(x)$ means that the algorithm $\cA$ outputs $y$ on input $x$.
If $\cA$ is a classical probabilistic or quantum algorithm that takes $x$ as input and outputs bit strings,
we often mean $\cA(x)$ by the output probability distribution of $\cA$ on input $x$.
When $\cA$ is a classical probabilistic algorithm, $y=\cA(x;r)$ means that the output of $\cA$ is $y$ if it runs on input $x$ and with the random seed $r$.
For two probability distributions $P\coloneqq\{p_i\}_i$ and $Q\coloneqq\{q_i\}_i$, 
$\SD(Q,P)\coloneqq\frac{1}{2}\sum_i|p_i-q_i|$ is their statistical distance.
$\TD(\rho,\sigma)$ is the trace distance between two states $\rho$ and $\sigma$. For pure states $\ket{\phi}$ and $\ket{\psi}$ we sometimes abuse notation and use $\TD(\ket{\phi}, \ket{\psi})$ to refer to $\TD(\ket{\phi}\!\bra{\phi}, \ket{\psi}\!\bra{\psi})$.

\subsection{Cryptography}

\begin{definition}[Security Game for iO for Oracle-Aided Classical Circuits]
\label{def:io-game}
\mor{It would be better to syncronyze the notaions with \cref{def:oracle_set} where $\mathsf{Obf}_\lambda$ was introduced and $\mathsf{Obf}(C,r)$ not $\mathsf{Obf}(C;r)$}
For $b\in \bin, \lambda\in \N, \obf:\bin^\lambda \times \bin^\lambda \rightarrow \bin^*, f\in \bin^\lambda \rightarrow \bin^*$, and an adversary $\cA$, define the game $\cG_{b,\lambda}^{\obf, f}(\cA)$ as:

 \begin{enumerate}
            \item $(C_0,C_1)\leftarrow A(1^\secp)$ where $C_0$ and $C_1$ are oracle-aided classical circuits.
            \item If $C^f_0$ is not functionally equivalent to $C^f_1$, or $C_1\notin\{0,1\}^{\lambda}$, or $C_0\notin\{0,1\}^{\lambda}$, the game outputs $\bot$. 
            \item $\obfC\gets\obf(C_b;r)$ for $r\leftarrow \bin^\lambda$.
            \item $b'\gets \cA(\obfC)$. 
            The game outputs $b'$.
        \end{enumerate}
    
\end{definition}
\begin{definition}[Indistinguishability Obfuscation for Oracle-Aided Classical Circuits \cite{JACM:BGIRSVY12}]
\label{def:oracle_iO}
\mor{It would be better to syncronyze the notaions with \cref{def:oracle_set} where $\mathsf{Obf}_\lambda$ was introduced and $\mathsf{Obf}(C,r)$ not $\mathsf{Obf}(C;r)$}
    For an oracle $\cO$, a pair of oracle-aided QPT algorithms $(\obf,\eval)$ is an indistinguishability obfuscator for $\cO$-aided classical circuits
    if the following conditions are satisfied.
    \begin{itemize}
        \item \textbf{Correctness:} For all $\secp\in\N$ and for all oracle-aided classical circuit $C\in\bit^\secp$, 
        \begin{align}
            \forall x,r,~ \eval^\cO(\obf^\cO(C;r),x)=C^{\cO}(x).
        \end{align}
        \item \textbf{Security:} 
        For any oracle-aided QPT adversary $\cA$,
        \begin{align}
            \left| \Pr[1\gets\cG^{\obf^\cO, \cO}_{0,\secp}(\cA^\cO)] - \Pr[1\gets\cG^{\obf^\cO, \cO}_{1,\secp}(\cA^\cO)] \right| \le \negl(\secp).
        \end{align}
        where $\cG$ is the game from \cref{def:io-game}.
    \end{itemize}
\end{definition}
\kabir{do we need to define more things?}

\subsection{Non-Collapsing Measurements and $\mathsf{PDQP}$}
\begin{definition}[Non-Collapsing Measurement Oracle $\cQ$ \cite{ITCS:ABFL16}] 
\label{def:noncollapsingmeasurementoracle}
    A non-collapsing measurement oracle $\cQ$ is an oracle that behaves as follows:
    \begin{enumerate}
        \item Take (a classical description of) a quantum circuit $C=(U_1,M_1,...,U_\tau,M_\tau)$ and an integer $\ell>0$ as input. 
        Here each $U_i$ is a unitary operator on $\ell$ qubits and each $M_i$ is a computational-basis projective measurement on $m_i$ qubits such that $0\le m_i\le\ell$. (When $m_i=0$, this means that no measurement is done.)
        \item Let $|\psi_0\rangle:=|0^\ell\rangle$. 
        Run $C$ on input $|\psi_0\rangle$,
        and obtain $(u_1,...,u_T)$, where
        $u_t\in\bit^{m_t}$ denotes the outcome of the measurement $M_t$ for each $t\in[T]$.
        Let $\tau_t:=(u_1,...,u_t)$.
        For each $t\in [T]$, 
        let $|\psi_t^{\tau_t}\rangle$ be the (normalized) post-measurement state immediately after the measurement $M_t$, i.e., 
        \begin{align}
            |\psi_t^{\tau_t}\rangle := \frac{ (|u_t\rangle\langle u_t|\otimes I)U_t|\psi_{t-1}^{\tau_{t-1}}\rangle }{ \sqrt{ \langle \psi_{t-1}^{\tau_{t-1}}|U_t^\dagger (|u_t\rangle\langle u_t|\otimes I)U_t|\psi_{t-1}^{\tau_{t-1}}\rangle } }.
        \end{align}
        \item 
        For each $t\in[T]$, 
        sample 
        $v_t\in\bit^\ell$ with probability
        $|\langle v_t|\psi_t^{\tau_t}\rangle|^2$.
        \item Output $(v_1,...,v_T)$.
    \end{enumerate}
\end{definition}

\begin{definition}[PDQP \cite{ITCS:ABFL16}]
    A language $L$ is in $\mathsf{PDQP}$ if there exists a polynomial-time classical deterministic Turing machine $R$ with a single query to a non-collapsing measurement oracle $\cQ$ such that
    \begin{itemize}
        \item For all $x\in L$, $\Pr[1\gets R^\cQ(x)]\ge \alpha(|x|)$,
        \item For all $x\notin L$, $\Pr[1\gets R^\cQ(x)]\le \beta(|x|)$,
    \end{itemize}
    where $\alpha,\beta$ are functions such that $\alpha(|x|)-\beta(|x|)\ge1/\poly(|x|)$.\footnote{Like $\mathsf{BQP}$, we can amplify the gap.}
\end{definition}

\begin{definition}[Sampling Problems~\cite{Aar14,ITCS:AarBuhKre24}]
\label{def:Samplingproblems}
A (polynomially-bounded) sampling problem is a collection $\{D_x\}_{x\in\bit^*}$ 
of probability distributions, 
where $D_x$ is a distribution over $\bit^{p(|x|)}$, for some fixed polynomial $p$.
\end{definition}

The sampling complexity class, $\mathsf{SampBQP}$, is defined as follows.
\begin{definition}[$\mathsf{SampBQP}$~\cite{Aar14,ITCS:AarBuhKre24}] 
\label{def:SampBQP}
$\mathsf{SampBQP}$ is the class of (polynomially-bounded) sampling problems 
$\{D_x\}_{x\in\bit^*}$ for which there exists a QPT algorithm $\cB$ such that for all $x$ and all $\epsilon>0$, 
$\SD(\cB(x,1^{\lfloor 1/\epsilon \rfloor}),D_x) \le\epsilon$, 
where $\cB(x,1^{\lfloor 1/\epsilon\rfloor})$ is the output probability distribution of 
$\cB$ on input $(x, 1^{\lfloor 1/\epsilon\rfloor})$. 
\end{definition}

We define $\mathsf{SampPDQP}$ as follows.
\begin{definition}[$\mathsf{SampPDQP}$~\cite{MorShiYam_PDQP}] 
\label{def:SampBQP}
$\mathsf{SampPDQP}$ is the class of (polynomially-bounded) sampling problems 
$\{D_x\}_{x\in\bit^*}$ for which there exists a classical deterministic polynomial-time 
algorithm $\cB$ that makes a single query to the non-collapsing measurement oracle $\cQ$ such that for all $x$ and all $\epsilon>0$, 
$\SD(\cB(x,1^{\lfloor 1/\epsilon \rfloor}),D_x) \le\epsilon$, 
where $\cB(x,1^{\lfloor 1/\epsilon\rfloor})$ is the output probability distribution of 
$\cB$ on input $(x, 1^{\lfloor 1/\epsilon\rfloor})$. 
\end{definition}

\subsection{Useful Lemmas}
In this paper, we use the following lemmas.
\begin{lemma}\label{lem:bbbv}{\cite{BBBV}}
Let $A$ be a quantum algorithm running in time $T$ with oracle access to $\cO$. 
Let $\epsilon > 0$ and let $S \subseteq [1,T] \times \{0,1\}^n$ be a set of time-string pairs such that 
$
\sum_{(t,r) \in S} q_r(|\phi_t\rangle) \leq \epsilon,
$
where $q_r(|\phi_t\rangle)$ is the magnitude squared of $r$ in the superposition of query $t$. 
If we modify $\cO$ into an oracle $\cO'$ which answers each query $r$ at time $t$ by providing the same string $R$ (which has been independently sampled at random), then the Euclidean distance between the final states of $A$ when invoking $\cO$ and $\cO'$ is at most $\sqrt{T\epsilon}$.
\end{lemma}

\begin{lemma}
\label{lem:markov-tv}~\cite{ITCS:ABFL16}
Suppose that $\tau \geq 1$, and that $v = (v_0, \ldots, v_\tau)$ is a random variable governed by a Markov distribution.  
That is, for all $1 \leq i \leq \tau$, $v_i$ is independent of $v_0, \ldots, v_{i-2}$ conditioned on a particular value of $v_{i-1}$.  
Let $w = (w_0, \ldots, w_\tau)$ be another random variable governed by a Markov distribution.  
If $\SD(\cdot, \cdot)$ denotes the total variation distance between random variables, then
\begin{align}
\SD(v,w) \;\leq\; 2 \sum_{i=1}^\tau \SD\big( (v_{i-1},v_i), (w_{i-1},w_i) \big).
\end{align}
\end{lemma}

\begin{lemma}
\label{thm:trace-dist-and-euclidean-dist}
    Let $\ket{\psi}$ and $\ket{\phi}$ be two pure states such that $\|\ket{\psi} - \ket{\phi}\| \leq \epsilon$. Then 
    $
    \TD\left(\ket{\psi}\!\bra{\psi}, \ket{\phi}\!\bra{\phi}\right) \leq \epsilon.
    $
    Here, $\TD(\rho,\sigma)$ is the trace distance between two states $\rho$ and $\sigma$.
\end{lemma}
\begin{proof}
    We have that
    \begin{align}
        \epsilon^2 &\geq \|\ket{\psi} - \ket{\phi}\|^2 
        = (\bra{\psi} - \bra{\phi}) (\ket{\psi} - \ket{\phi}) 
        = 2 - (\langle \phi | \psi \rangle + \langle \phi | \psi \rangle) \\
        &= 2 - 2 \text{Re} (\langle \phi | \psi \rangle) 
        \geq 2 - 2 |\langle \phi | \psi \rangle|,
    \end{align}
    which can be rearranged to 
    $
        |\langle \phi | \psi \rangle| \geq 1 - \frac{\epsilon^2}{2}.
        $
    By the identity for trace distance of pure states,
    \begin{align}
        \TD(\ket{\psi}\!\bra{\psi}, \ket{\phi}\!\bra{\phi}) = \sqrt{1 - |\langle \psi | \phi \rangle|^2} 
        \le \sqrt{1 - \left(1 - \frac{\epsilon^2}{2} \right)^2} 
        = \sqrt{\epsilon^2 - \frac{\epsilon^4}{4}} 
        \leq \epsilon.
    \end{align}
    This completes the proof.
\end{proof}

\begin{lemma}\label{lem:thetaED-to-TD}
For pure states $\ket{\psi}$ and $\ket{\psi'}$, if 
$\min_\theta \| \ket{\psi} - e^{i\theta}\ket{\psi'} \| \geq\delta$, 
then $\TD(\ket{\psi},\ket{\psi'}) \geq\delta / \sqrt{2}$.
\end{lemma}

\begin{proof}
\begin{align}
  \|\ket{\psi} - e^{i\theta}\ket{\psi'}\| 
    &= \sqrt{2 - 2 \mathsf{Re} \bigl( e^{i\theta} \langle \psi | \psi' \rangle \bigr)} \\
  \min_\theta \|\ket{\psi} - e^{i\theta}\ket{\psi'}\| 
    &= \sqrt{2 - 2|\langle \psi | \psi' \rangle|} \geq\delta .
\end{align}
Thus
$
  1 - \tfrac{\delta^2}{2} \geq|\langle \psi | \psi' \rangle| .
  $
Note that since the RHS is positive, $\delta^2 \leq2$. We can use the bound on the inner product to bound the trace distance.
\begin{align}
  \TD(\ket{\psi},\ket{\psi'}) 
    = \sqrt{1 - |\langle \psi | \psi' \rangle|^2} 
    \geq\sqrt{1 - \left(1 - \tfrac{\delta^2}{2}\right)^2} 
    = \sqrt{\delta^2 - \tfrac{\delta^4}{4}} 
    \geq\delta \sqrt{1 - \tfrac{\delta^2}{4}} 
    \geq\delta / \sqrt{2}.
\end{align}
\end{proof}
\begin{lemma}
\label{lem:ED-to-SD}
For a distribution $\cD$, define a state
    $\ket{\cD}\coloneqq \sum_x \sqrt{\Prr_{\cD}[x]} \ket{x}$,
    where $\Prr_{\cD}[x]$ is the probability that $x$ is sampled from $\cD$.
For any two distributions $\cD$ and $\cD'$,
$
\|\ket{\cD} - \ket{\cD'}\| \leq \sqrt{2\cdot\SD(\cD,\cD')}.
$
\end{lemma}
\begin{proof}
    \begin{align}
        \|\ket{\cD} - \ket{\cD'}\| 
        &= \left\|\sum_x \left(\sqrt{{\rm Pr}_{\cD}[x]} - \sqrt{{\rm Pr}_{\cD'}[x]}\right) \ket{x}\right\|
        = \sqrt{\sum_x \left(\sqrt{{\rm Pr}_{\cD}[x]} - \sqrt{{\rm Pr}_{\cD'}[x]}\right)^2}\\
        &\leq \sqrt{\sum_x \left|\left(\sqrt{{\rm Pr}_{\cD}[x]} - \sqrt{{\rm Pr}_{\cD'}[x]}\right)\left(\sqrt{{\rm Pr}_{\cD}[x]} + \sqrt{{\rm Pr}_{\cD'}[x]}\right)\right|}\\
        &=\sqrt{\sum_x \left|{\rm Pr}_{\cD}[x] - {\rm Pr}_{\cD'}[x]\right|}
        =\sqrt{2\cdot \SD(\cD,\cD')}.
    \end{align}
\end{proof}

\section{Quantum Oracle Queries}
\label{sec:qq}
In this section, we introduce some helpful techniques to analyse algorithms that make quantum queries to oracles.

In \cref{sec:def_query} we review the standard notion of quantum queries to classical oracles. In \cref{sec:def_comp} we show how quantum access to classical oracles whose truth tables are drawn randomly from a product distribution can be simulated using query access to a particular kind of unitary (that we call a compression unitary) by using the compressed oracle technique from \cite{C:Zhandry19}.
In \cref{sec:OW2H}, we show how for (a slightly generalized version of) compression unitaries, distinguishing between two unitaries using quantum queries also allows for finding inputs on which the two unitaries  behave differently. This is proved via a hybrid method, and will be used extensively to indistinguishably switch between two unitaries. We also show a similar theorem for random classical oracles.

\subsection{Defining Quantum Oracle Queries}
\label{sec:def_query}
We define a $q$-query algorithm as an algorithm that makes at most $q$ quantum queries, each of length at most $q$. We also assume that the output of a $q$-query algorithm is of length at most $q$. 

Let $F$ be a function mapping bitstrings to bitstrings. 
Quantum queries to $F$ may be modeled as follows. The querier prepares registers 
$\reg{X}\reg{Y}\reg{Z}$ where $\reg{X}$ contains the query string, $\reg{Y}$ is 
the response register, and $\reg{Z}$ may contain auxiliary information. 
A query consists of applying the unitary
\begin{align}
  \sum_x \ket{x}\!\bra{x}_{\reg{X}} \otimes X^{F(x)}_{\reg{Y}} .
\end{align}
Equivalently, we may initialize a register $\reg{D} = \bigotimes_{x \in \{0,1\}^*} \reg{D}_x$ 
to the truth table of $F$. Formally, $\reg{D}$ is initialized to 
$\bigotimes_x \ket{F(x)}_{\reg{D}_x}$. A query may now be modeled as
\begin{align}
  \sum_x \ket{x}\!\bra{x}_{\reg{X}} \otimes \mathsf{CNOT}_{\reg{D}_x \reg{Y}} ,
\end{align}
where $\mathsf{CNOT}_{\reg{D}_x \reg{Y}}$ maps $\ket{s_1}_{\reg{D}_x}\ket{s_2}_\reg{Y}$ to $\ket{s_1}_{\reg{D}_x}\ket{s_2 \oplus s_1}_\reg{Y}$. Note that the querier does not have direct access to $\reg{D}$.

\subsection{Compressed Oracles for General Product Distributions}
\label{sec:def_comp}
Let $\{\cD_x\}_{x \in \{0,1\}^*}$ be a collection of distributions over strings. \takashi{Is the length of $x$ unbounded? In that case, are we considering infinite dimensional space below?} \kabir{since we only care about query efficient it doesn't matter what the bound is, but in all our use cases there is a bound.}
Let $\cD$ be the product distribution induced by these distributions, i.e., $\cD := \bigotimes_x \cD_x$. 
We say an oracle $\cO$ is sampled from $\cD$ when the $x$-th row of its truth table 
is sampled independently from $\cD_x$. In other words, the truth table of $\cO$ is sampled from $\cD$.

We define the notion of the compression unitary for a product distribution as follows.
\begin{definition}[Compression Unitary]
\label{def:comp_U}
    Let $\{\cD_x\}_{x\in\bit^*}$ be a collection of distributions over strings.
    Let $\cD:=\bigotimes_x \cD_x$ be the product distribution induced by $\{\cD_x\}_{x\in\bit^*}$.
    Then, the compression unitary $U$ for $\cD$ is 
    \begin{align}
        U:= \sum_x |x\rangle\langle x| \otimes U_x,
    \end{align}
    where 
    \begin{align}
        U_x := |\cD_x\rangle\langle\bot| _{\regD_x} + |\bot\rangle\langle \cD_x|_{\regD_x} + (I-|\bot\rangle\langle\bot|-|\cD_x\rangle\langle\cD_x|)_{\regD_x}.
    \end{align}
  Here,  
  $\ket{\cD_x} :=\sum_z \sqrt{{\rm Pr}_{\cD_x}[z]} \ket{z}$
and $\Pr_{\cD_x}[z]$ is the probability that $z$ is sampled from $\cD_x$.
\end{definition}

Next, we show that quantum queries to oracles drawn from a product distribution can be simulated using quantum access to the corresponding compression unitary using a technique called the compressed oracle, which was first introduced in \cite{C:Zhandry19}.
\begin{definition}[Compressed Oracle]
The compressed oracle $\mathsf{CStO}_\cD$ is a method by which quantum queries to a randomly sampled classical oracle can be statefully simulated. 
    Let $\{\cD_x\}_{x\in\bit^*}$ be a collection of distributions over strings.
    Let $\cD:=\bigotimes_x \cD_x$ be the product distribution induced by $\{\cD_x\}_{x\in\bit^*}$.
The database register $\regD:=\bigotimes_x\regD_x$ is private to the simulation and is initialized as $|\bot\rangle_\regD:=\bigotimes_x|\bot\rangle_{\regD_x}$.
    A query to the compressed oracle $\mathsf{CStO}_\cD$ is implemented by applying the following unitary to the querier's query register $\reg{X}$, response register $\reg{Y}$, and the private database register $\reg{D}$: 
    \begin{align}
        \sum_x \ket{x}\!\bra{x}_{\reg{X}} \otimes (U_x \circ \mathsf{CNOT}_{\reg{D}_x \reg{Y}} \circ U_x),
    \end{align}
    where $U = \sum_x\ket{x}\!\bra{x}\otimes U_x$ is the compression unitary for $\cD$ as defined in \cref{def:comp_U}.
\end{definition}
We note that $\mathsf{CStO}_\cD$ can be implemented using $U$ in a query-efficient way. In particular, the simulation makes two queries to $U$ per $\mathsf{CStO}_\cD$ query.

\begin{theorem}
\label{thm:comp-oracle}
For any adversary $\adv$,
\begin{align}
  \Pr[\adv^{\mathsf{CStO}_\cD} = 1] = \Prr_{O \leftarrow \cD}[\adv^O = 1].
\end{align}
\end{theorem}
\begin{proof} Follows from the proof of Lemma 4 in \cite{C:Zhandry19} with a \kabir{no "a"?} minimal modification.
\end{proof}

\subsection{Bounds on Distinguishing Oracles}
Here we show some bounds on distinguishing oracles. First we look at any pair of unitaries that take as input a register $\reg{X}$ and a target register, and on any computational basis state $\ket{x}_\reg{X}$ apply an operation that swaps between $\ket{\bot}$ and some state $\ket{\phi_x}$ to the target register. We note that all compression unitaries (\cref{def:comp_U}) are of this form. We show that the ability to distinguish between unitaries of this form allows for finding strings $x$ where the corresponding $\ket{\phi_x}$ states for the unitaries are far. This is proved via a hybrid argument, and will be used extensively to indistinguishably switch between two unitaries. 

\label{sec:OW2H}
\begin{theorem}
\label{thm:ow2h-comp}
For the collection of states $\{\ket{\varphi_x}\}_x$ and $\{\ket{\varphi'_x}\}_x$ define
\begin{align}
    U_x &:= \ket{\varphi_x}\!\bra{\bot}_{\reg{D}_x} + \ket{\bot}\!\bra{\varphi_x}_{\reg{D}_x}
  + \big(\bbI - \ket{\bot}\!\bra{\bot}_{\reg{D}_x} - \ket{\varphi_x}\!\bra{\varphi_x}_{\reg{D}_x}\big)\label{U_x}\\
  U'_x &:= \ket{\varphi'_x}\!\bra{\bot}_{\reg{D}_x} + \ket{\bot}\!\bra{\varphi'_x}_{\reg{D}_x}
  + \big(\bbI - \ket{\bot}\!\bra{\bot}_{\reg{D}_x} - \ket{\varphi'_x}\!\bra{\varphi'_x}_{\reg{D}_x}\big).\label{U_x'}
\end{align}
Also define
\begin{align}
  U &:= \sum_x \ket{x}\!\bra{x}_{\reg{X}} \otimes U_x\\
  U' &:= \sum_x \ket{x}\!\bra{x}_{\reg{X}} \otimes U'_x.
\end{align}
For any $q$-query adversary $\adv$, 
let $\advB$ be the algorithm that samples $i \gets [q]$ and runs $\adv$ up to just before 
the $i$-th query, then measures the query register in the computational basis 
and returns the measurement output. Let $\Delta := \|\ket{\psi}- \ket{\psi'}\|$ 
where $\ket{\psi}$ and $\ket{\psi'}$ are the purified final states of $\adv^U$ and $\adv^{U'}$ 
respectively. Then
\begin{align}
  \Exp_{x \leftarrow \cB^U} 
  \big[ \|\ket{\varphi_x} - \ket{\varphi'_x}\| \big] 
  \geq\frac{\Delta^2}{32 q^2}.
\end{align}
\end{theorem}

\begin{proof}
Consider the following sequence of hybrids $\{\cH_i\}_{i\in[0,q]}$, where $\cH_i$ consists of running $\adv$, answering the first $i$ queries using $U$, and answering the remaining queries using $U'$. 
Let the (purified) final state of $\adv$ in $\cH_i$ be $\ket{\psi_i}$. Note that $\ket{\psi_0} = \ket{\psi'}$ and $\ket{\psi_q} = \ket{\psi}$.
By the triangle inequality,
\begin{align}
  \sum_{i=1}^q \|\ket{\psi_i} - \ket{\psi_{i-1}}\| \geq\Delta.
\end{align}
Define $\varepsilon_i := \| \ket{\psi_i} - \ket{\psi_{i-1}} \|$. 
The state just before the $i$-th query in $\cH_i$ and $\cH_{i-1}$ is identical since both hybrids are identical until the $i$-th query. Let this state be $\ket{\widetilde{\psi_i}}$. If we write the operation of $\adv$ between queries as the unitary $A$ and consider the initial state to be $\ket{0}$ without loss of generality then
\begin{align}
  \ket{\widetilde{\psi}_i} = (AU)^{i-1} A \ket{0}.
\end{align}
The final state in $\cH_i$ is $(AU')^{q-i} A (U \ket{\widetilde{\psi}_i})$ 
and in $\cH_{i-1}$ it is $(AU')^{q-i} A (U' \ket{\widetilde{\psi}_i})$ .  Thus,

\begin{align}
  \varepsilon_i &= \|\ket{\psi_i} - \ket{\psi_{i-1}}\|  
  = \|(AU')^{q-i} A (U \ket{\widetilde{\psi}_i}) - (AU')^{q-i} A (U' \ket{\widetilde{\psi}_i})\|\\
  &= \|U \ket{\widetilde{\psi}_i} - U' \ket{\widetilde{\psi}_i}\|
  = \|(U - U') \ket{\widetilde{\psi}_i}\|.
\end{align}
We can decompose $\ket{\widetilde{\psi}_i}$ into
\begin{align}
  \ket{\widetilde{\psi}_i} = \sum_x \sqrt{p_x}\ket{x}\ket{\phi_x},
\end{align}
where the first register is the query register and $p_x$ is the probability that measuring the query register in the computational basis returns $x$. Applying $U$ gives
\begin{align}
  U\ket{\widetilde{\psi}_i} = \sum_x \sqrt{p_x}\ket{x} \otimes U_x \ket{\phi_x},
\end{align}
while applying $U'$ gives
\begin{align}
  U'\ket{\widetilde{\psi}_i} = \sum_x \sqrt{p_x}\ket{x} \otimes U'_x \ket{\phi_x}.
\end{align}
Thus
\begin{align}
 \label{eq:ow2h-helper}   
  \varepsilon_i^2 = \|(U - U')\ket{\widetilde{\psi}_i}\|^2 
  = \sum_x p_x \| (U_x - U'_x)\ket{\phi_x} \|^2.
\end{align}
Let $\ket{\xi_x} := \ket{\varphi_x} - \ket{\varphi'_x}$ be an un-normalized state. 
By the definition of $U_x$ and $U'_x$, \cref{U_x,U_x'}, and
taking their difference,
\begin{align}
  U_x - U'_x &= \ket{\xi_x}\!\bra{\bot} + \ket{\bot}\!\bra{\xi_x} + (\ket{\varphi'_x}\!\bra{\varphi'_x} - \ket{\varphi_x}\!\bra{\varphi_x})\\
  &= \ket{\xi_x}\!\bra{\bot} + \ket{\bot}\!\bra{\xi_x} + \Big(\ket{\varphi'_x}(\bra{\varphi_x} - \bra{\xi_x}) - (\ket{\varphi'_x} + \ket{\xi_x})\bra{\varphi_x}\Big)\\
  &= \ket{\xi_x}\!\bra{\bot} + \ket{\bot}\!\bra{\xi_x} -\ket{\varphi'_x}\!\bra{\xi_x} - \ket{\xi_x}\!\bra{\varphi_x}.
\end{align}
Applying the operator to $\ket{\phi_x}$ and noting that $\ket{\varphi_x}$ and $\ket{\varphi'_x}$ are not supported on $\bot$,
\begin{align}
  \|(U_x - U'_x)\ket{\phi_x}\| 
    &= \left\| \left(\ket{\xi_x}\!\bra{\bot} + \ket{\bot}\!\bra{\xi_x} -\ket{\varphi'_x}\!\bra{\xi_x} - \ket{\xi_x}\!\bra{\varphi_x}\right)\ket{\phi_x}\right\| \\
    &=\left\| \left(\sqrt{2}\cdot\ket{\xi_x}\tfrac{\bra{\bot} - \bra{\varphi_x}}{\sqrt{2}}  + \sqrt{2}\cdot\tfrac{\ket{\bot} - \ket{\varphi'_x}}{\sqrt{2}}\bra{\xi_x}\right)\ket{\phi_x}\right\| \\
    &\leq2\sqrt{2} \|\ket{\xi_x}\| 
    =2\sqrt{2} \|\ket{\varphi_x} - \ket{\varphi'_x}\|.
\end{align}
Plugging this bound back into \eqref{eq:ow2h-helper},
\begin{align}
\varepsilon_i^2 \leq\sum_x p_x \cdot 8 \cdot \|\ket{\varphi_x}-\ket{\varphi'_x}\|^2
   \leq 8 \cdot \Exp_x[ \|\ket{\varphi_x}-\ket{\varphi'_x}\|^2] ,
\end{align}
where $x$ is sampled by measuring the query register of $\ket{\widetilde{\psi}_i}$, which is exactly the distribution output by $\cB^U$ conditioned on sampling $i$.
Therefore,
\begin{align}
\Exp_{x \leftarrow \cB^U} \left[\|\ket{\varphi_x}-\ket{\varphi'_x}\| \right]
   &\geq \frac{1}{4}\cdot\Exp_{x \leftarrow \cB^U} \left[\|\ket{\varphi_x}-\ket{\varphi'_x}\|^2 \right]
   \geq\frac{1}{4}\cdot\Exp_{i \leftarrow [q]} \left[ \frac{\varepsilon_i^2}{8} \right] \\
   &\geq \frac{1}{32}\cdot \Exp_{i \leftarrow [q]}\left[\varepsilon_i \right]^2
   =\frac{1}{32}\left( \frac{\sum_i \varepsilon_i}{q}\right)^2
   \ge\frac{\Delta^2}{32q^2}.
\end{align}
\end{proof}
Next, we show a similar theorem for random classical oracles. More specifically, distinguishing between oracles drawn from two different product distributions allows for finding indices for which the corresponding sub-distributions are far.
\begin{theorem}
\label{thm:ow2h-dist}
Let $\{\cD_x\}_{x \in \{0,1\}^*}$ and
$\{\cD_x'\}_{x \in \{0,1\}^*}$ 
be collections of distributions over strings. 
Let $\cD\coloneq\bigotimes_x \cD_x$ and $\cD'\coloneqq\bigotimes_x \cD_x'$ be the product distributions induced by them.
Let $U$ and $U'$ 
be the corresponding compression unitaries for $\cD$ and $\cD'$, respectively. (For the definition of compression unitaries, see \cref{def:comp_U}.) For any $q$-query adversary $\adv$, 
let $\advB$ be the algorithm that samples $i \gets [q]$ and runs $\adv$ up to just before 
the $i$-th query, then measures the query register in the computational basis 
and returns the measurement output. Let $\Delta := \|\ket{\psi}- \ket{\psi'}\|$ 
where $\ket{\psi}$ and $\ket{\psi'}$ are the purified final states of $\adv^U$ and $\adv^{U'}$ 
respectively. Then
\begin{align}
  \Exp_{x \leftarrow \cB^U} 
  \big[ \SD(\cD_x, \cD'_x) \big] 
  \geq\frac{\Delta^2}{16 q^2}.
\end{align}
\end{theorem}
\begin{proof}
    Identical to the proof of \cref{thm:ow2h-comp} except that after obtaining
    \begin{align}
    \|(U_x - U'_x)\ket{\phi_x}\| 
    \leq2\sqrt{2} \|\ket{\varphi_x} - \ket{\varphi'_x}\|
    \end{align}
    we use Lemma \ref{lem:ED-to-SD} to additionally obtain
    \begin{align}
    \|(U_x - U'_x)\ket{\phi_x}\| 
    \leq4 \sqrt{\SD(\cD_x,\cD'_x)}.
    \end{align}
\end{proof}

\section{Indistinguishability Obfuscation and One-Way Permutations Relative to OW2H-Compatible Oracles}
\label{sec:io}
In this section, we show the existence of a random classical oracle $S$ that allows us to instantiate iO and OWPs that remain secure in the presence of a class of unitary oracles we call "one-way to hiding (OW2H) compatible" unitaries. First, in \cref{sec:punc}, we define this class of unitaries.
Then in \cref{sec:def_oracle_iOWP}, we define the classical oracle $S$.
Finally, in \cref{subsec:io-owp-from-s}, we show that $S$ instantiates iO and OWP that remain secure in the presence of any OW2H-compatible unitaries.
\subsection{OW2H-Compatibility}
\label{sec:punc}
One-way to hiding theorems \cite{EC:Unruh12} relate the probability of distinguishing oracles to the probability of finding points where the oracles differ. We adapt this concept to define a class of oracles which satisfy a strong form of one-way to hiding. Specifically, for a set of classical oracles $\bbO$ and a collection of unitaries $\{W_O\}_{O\in\bbO}$, we will frequently need to argue that for $O_1, O_2 \in \bbO$, $(O_1, W_{O_1})$ and $(O_2, W_{O_2})$ are indistinguishable given query access, where $O_1$ and $O_2$ only differ on some hard to find inputs. One-way to hiding theorems in the literature readily allow showing claims of this sort in the absence of the second unitary oracle. We will say that $\{W_O\}_{O\in\bbO}$ is one-way to hiding (OW2H) compatible for $\bbO$ if such claims hold even in the presence of the unitary oracle. We also require $O$ to be computable by querying $W_O$ for the sake of notational convenience, and will often only provide access to $W_O$ instead of $(O, W_O)$. 

\begin{definition}[OW2H-Compatibility]
\label{def:punc}
Let $\bbO$ be a collection of classical oracles and let $W := \{W_O\}_{O \in \bbO}$ be a collection of unitaries. We say that $W$ is OW2H-compatible for $\bbO$ if:

\begin{enumerate}
    \item $O$ can be exactly computed by making a single query to $W_O$. \takashi{I'm wondering where we used this condition.}
    \item For all functions $q:\N\to\N$ there exists a $\poly(q(\lambda))$-query \alper{What is $q$ here, i think this should be something like: for every poly q there is a polynomial q' such that for a $q'(\lambda)$ query algorithm}\kabir{better?} algorithm $\cB$ and a constant $c \geq 0$ such that for all $O, O' \in \bbO$ and 
    for any $q(\lambda)$-query adversaries $\cA$,
    \begin{align}
        \left\| \ket{\psi} - \ket{\psi'} \right\| \leq \poly(q(\lambda)) \cdot \left(\Prr[ \cB^{W_O, \cA}(1^\lambda) \in T]\right)^{c},
    \end{align}
    where $\ket{\psi}$ and $\ket{\psi'}$ are the purified final states of $\cA^{W_O}(1^\lambda)$ and $\cA^{W_{O'}}(1^\lambda)$ respectively, and $T := \{ x : O(x) \neq O'(x) \}$. We will usually drop the notation for query access to $\cA$ from $\cB$ when obvious from context.
\end{enumerate}
\end{definition}
We will use the following property of OW2H-compatibility extensively.
\begin{theorem}
    \label{thm:ow2h-punc-full}
    Let $\bbO$ be a collection of functions and let $W := \{ W_O \}_{O \in \bbO}$ be a collection of unitaries where $W$ is OW2H-compatible for $\bbO$ with constant $c \in (0,1)$. Let $(O_1, O_2, z)$ be jointly distributed random variables where $O_1, O_2 \in \bbO$ and $z$ is some classical side information. 
    Then for all functions $q:\N\to\N$, for any $q(\lambda)$-query adversaries \alper{same as above} \kabir{fine?} $\cA$ that outputs a quantum state, and for all measurements $\mathcal{M}$ that may depend on $O_1, O_2$, and $z$, there exists a $\poly(q(\lambda))$-query adversary $\cB$ that outputs a bit string such that
\begin{align*}
&\Bigg| \Pr_{O_1, O_2, z}\Big[ \mathcal{M}(\cA^{{W_{O_1}}}(z, 1^\lambda)) = 1 \Big] - \Pr_{O_1, O_2, z}\Big[ \mathcal{M}(\cA^{{W_{O_2}}}(z, 1^\lambda)) = 1 \Big] \Bigg|\\ &\leq \poly(q) \cdot \Pr_{O_1, z} \Big[ \cB^{{W_{O_1}}}(z, 1^\lambda) \in T \Big]^{c},
\end{align*}
where $T := \{ x : O_1(x) \neq O_2(x) \}$
and $\mathcal{M}(\rho)=1$ means that the measurement $\mathcal{M}$ succeeds when applied to the state $\rho$.
\end{theorem}
\begin{proof}
    We drop $1^\lambda$ from the notation for inputs to algorithms and assume it is always provided \kabir{clunky?}.  Let $\cB$ be the algorithm provided by Definition $\ref{def:punc}$ for $q(\lambda)$-query adversaries. For any fixing of $O_1,O_2,z$ let $\varepsilon$ be defined as:
    \begin{align}
    \varepsilon := \left|\Pr\Big[ \mathcal{M}(\cA^{{W_{O_1}}}(z)) = 1 \Big] - \Pr\Big[ \mathcal{M}(\cA^{{W_{O_2}}}(z)) = 1 \Big]\right|.
    \end{align}
    Let $\ket{\psi_1}$ and $\ket{\psi_2}$ be the purified final states (before applying the measurement $\cM$) of $\cA^{{W^{O_1}}}$ and $\cA^{{W^{O_2}}}$, respectively. Then by the data processing inequality,
    $
        \TD(\ket{\psi_1}, \ket{\psi_2}) \geq \varepsilon.
        $
    By \cref{thm:trace-dist-and-euclidean-dist}, this implies
    $
        \|\ket{\psi_1}- \ket{\psi_2}\| \geq \epsilon.
        $
    By the definition of $\cB$,
    \begin{align}
    \varepsilon \leq \poly(q(\lambda))\cdot\left(\Pr[\cB^{{W_{O_1}}}(z) \in T]\right)^c.
    \end{align}
    Note that $\cB$ depends on $\cA$ but not on $\cM$. Taking the expectation over $(O_1,O_2,z)$, noting that $c \in (0,1)$, and applying Jensen's inequality gives the statement of the theorem.
\end{proof}

\subsection{Definition of Oracle $S$}
\label{sec:def_oracle_iOWP}
We define a random oracle $S$ as follows:
\begin{definition}[Oracle $S$]
\label{def:oracle_set}
    The oracle $S:=(f,\mathsf{Obf},\Eval^{f})$ consists of the following three parts:
    \begin{itemize}
        \item For all $\secp\in\N$ and for all $x\in\bit^\secp$, $f(x)=f_\lambda(x)$ where $f_\lambda$ is a random permutation on $\lambda$-bit strings.
        \item For all $\lambda\in\N$, for all $C\in\bin^{\lambda}$, and for all $r\in\bin^{\lambda}$, $\obf(C,r)=\obf_\lambda(C,r)$ where $\obf_\lambda : \bin^\lambda \times \bin^\lambda \rightarrow \bin^{3\lambda}$ is a random injective function.
        \item For any function $f'$, $\eval^{f'}(\obfC,x)$ performs the following: 
            \begin{enumerate}
                \item If $\obfC\notin \mathsf{Image}(\obf)$, return $\bot$.
                \item Let $(C,r):=\obf^{-1}(\obfC)$. Parse $C$ as an oracle-aided (classical) circuit.
                \item If $|x|$ is not consistent with the input size of $C$, then return $\bot$. Otherwise, return $C^{f'}(x)$. 
            \end{enumerate}
            Note that $\eval^f$ is defined in terms of $\obf$, i.e. $\eval^f$ "knows" the truth table of $\obf$.
    \end{itemize}
    We will use $\bbS$ to represent the set of all possible oracles with the same interface as $S$. \kabir{does this line sound fine?}
    \mor{The word "interface" is slightly not clear. If possible, could you make it clearer?}
\end{definition}

\subsection{Constructing One-Way Permutations and Indistinguishability Obfuscation using $S$}
\label{subsec:io-owp-from-s}
In this section we use the oracle $S$ to construct a candidate OWP and a candidate iO scheme. We will then show a strong security property. Let $\bbS$ be the set of all possible oracles with the same interface as $S$, and let $W = \{W^O\}_{O\in\bbS}$ be a collection of unitaries that are OW2H-compatible for $\bbS$. We will show that the constructed primitives are secure against query-bounded quantum adversaries that have access to both $S$ and $W^S$. Since queries to $S$ can be answered exactly using access to $W^S$, it suffices to only consider access to $W^S$. 
\subsubsection{OWPs Secure Against Any OW2H-Compatible Unitary Collection}
We show that $S=(f,\obf,\eval^f)$ realizes a OWP secure against any OW2H-compatible unitary collection. The permutation is simply the function $f$. The following theorem proves security.
\begin{theorem}
\label{thm:owp}
Let $S = (f,\obf,\eval^f)$ be the random oracle defined in \cref{def:oracle_set}, and let $\bbS$ be the set of all possible oracles with the same interface as $S$. Let $W = \{W^O\}_{O\in\bbS}$ \kabir{can also say $W^S$ if that is clearer} be a collection of unitaries that are OW2H-compatible with $\bbS$ (\cref{def:punc}) with constant $c\in(0,1)$. 
For all functions $q:\N\to\N$ such that $q(\secp)\leq2^{o(\lambda)}$, 
for all $q(\lambda)$-query adversaries $\cA$, 
and for all sufficiently large $\lambda$,
\begin{align}
\Prr_{S,x\leftarrow\bin^{\lambda}}\Big[\cA^{W^S}\big(f(x)\big)=x\Big]\leq\frac{1}{2^{c^2\lambda/2}}.
\label{onewaypermutationsecurity}
\end{align}
\mor{$S\gets\bbS,x\gets\bit^\lambda$?}

\end{theorem}

Before proving \cref{thm:owp} we will show that for a fixed $\obf$ and for any choice of $(f', f'')$, if $(f', \obf, \eval^{f'})$ and $(f'', \obf, \eval^{f''})$ differ on some input $y$ (i.e. at least one of the corresponding oracles in the tuples differ on input $y$), then $y$ can be used to find $y'$ where $f'$ and $f''$ differ with noticeable probability.
\begin{lemma}
    \label{clm:owp-find}
For any $f',f'',\obf$ and $y$ such that $f'(y)\ne f''(y)$ or $\eval^{f'}(y)\ne \eval^{f''}(y)$,
\begin{align}
\Pr\Big[f'(y')\ne f''(y') :y'\leftarrow \mathsf{Find}^{f'}(y)\Big]\ge\frac{1}{2|y|},
\end{align}
where $\mathsf{Find}^{f'}(y)$ does the following: \takashi{It may be useful to mention that $\mathsf{Find}$ ``knows" the full truth table of $\obf$. (Since it only has oracle $f'$ one may misunderstand that it is independent of $\obf$.)}\kabir{see bottom of the lemma now}\takashi{looks good}
\begin{itemize}
\item With probability $1/2$ output $y$.
\item With probability $1/2$, interpret $y$ as $(\obfC, \widetilde{x})$ and do the following:
\begin{itemize}
\item If $\obfC\notin \mathsf{Image}(\obf)$, return $\bot$.
\item Let $(C,r):=\obf^{-1}(\obfC)$. Parse $C$ as an oracle-aided (classical) circuit.
\item Run $C^{f'}(\widetilde{x})$  and keep track of the $f'$-queries made by $C$.
\item Randomly choose one of them and output it.
\end{itemize}
\end{itemize}
\mor{I think $f'(y')\neq f''(y')$ is not necessarily satisfied if $y$ satisfies $f'(y)=f''(y)$ and $y=(\obfC,\widetilde{x})$ satisfies $\obfC\notin\mathsf{Image}(\obf)$.}
Crucially note that $\mathsf{Find}^{(\cdot)}(y)$ makes at most $|y|$ queries, and "knows" the full truth table of the fixed $\obf$.
\end{lemma}

\begin{proof}[Proof of \cref{clm:owp-find}]
If $\eval^{f'}(y)= \eval^{f''}(y)$ then it must be the case that $f'(y) \neq f''(y)$ already. If $\eval^{f'}(y)\ne \eval^{f''}(y)$, then $C^{f'}$ must query some $y'$ such that $f'(y')\ne f''(y')$. The claim follows from observing that $C$ is a circuit of size at most $|y|$ and so makes at most $|y|$ queries.
\end{proof}
\noindent 

Now we show \cref{thm:owp}.

\begin{proof}[Proof of \cref{thm:owp}]
We will show that for any fixing of $\obf$, for any function $q:\N\rightarrow\N$ such that $q(\lambda) \leq 2^{o(\lambda)}$, for any $q(\lambda)$-query adversary $\cB$, for all large enough $\lambda$,
\begin{align}
\Prr_{f,{x\leftarrow\bin^\lambda}}\Big[\cB^{W^{S}}\big(f(x)\big)=x\Big]\leq\frac{1}{2^{c^2\lambda/2}}.
\end{align}
This suffices to show \cref{onewaypermutationsecurity}.
For the remainder of the proof, fix any $\obf$. 

\noindent For any set $\bbX\subseteq\bit^*$ and all $x^*\in\{0,1\}^*$, we define
\begin{align}
f_{\bbX}(x^*)=
\begin{cases}
f(x^*) & \text{if } x^*\notin\bbX,\\
\bot & \text{if } x^*\in\bbX.
\end{cases}
\end{align}
Let $x'$ be sampled uniformly at random from $\{0,1\}^\lambda$. Define the random variables 
\begin{align}
S_1&:=(f_{\{x'\}},\obf,\eval^{f_{\{x'\}}})\\
S_2&:=(f_{\{x,x'\}},\obf,\eval^{f_{\{x,x'\}}})\\
S_3&:=(f_{\{x\}},\obf,\eval^{f_{\{x\}}})\\
S_4&:=(f,\obf,\eval^{f}).
\end{align}
Note that the $S_i$s are correlated with the choice of $x$ and $x'$ which are both sampled uniformly at random as stated above.
Also note that $S_4=S$, and therefore our goal is to show that 
\begin{align}
\Prr_{f,x}\Big[\cB^{W^{S_4}}\big(f(x)\big)=x\Big]\leq\frac{1}{2^{c^2\lambda/2}}.
\end{align}
We will show this by first showing that $x$ is hard to find when given input $f(x')$ and query access to $W^{S_1}$, then showing by hybrid argument that we can switch from giving the adversary input $f(x')$ and query access to $W^{S_1}$ to giving input $f(x)$ and query access to $W^{S_4}$ without significantly increasing the probability of finding $x$.
First, when using the oracle $S_1$, we can easily bound the probability that $\cB$ outputs $x$.

\begin{MyClaim}
\label{clm:owp-1}
For all functions $q:\N\to\N$ and all $q(\lambda)$-query adversaries $\cB$,
\begin{align}
\Pr_{f,x,x'}\big[\cB^{W^{S_1}}(f(x'))=x\big]\le \frac{1}{2^\lambda}.
\end{align}
\end{MyClaim}
\begin{proof}[Proof of \cref{clm:owp-1}]
Follows from the observation that $x$ is sampled independent of the adversary’s view.
\end{proof}
Next, we bound the probability that $\cB$ outputs $x$ when the oracle is $S_2$.
\begin{MyClaim}
\label{clm:owp-2}
For all functions $q:\N\to\N$ and all $q(\lambda)$-query adversaries $\cB$, and for all sufficiently large $\lambda$,
\begin{align}
\Pr_{f,x,x'}\big[\cB^{W^{S_2}}(f(x'))=x\big]\le \frac{\poly(q(\lambda))}{2^{c\lambda}}.
\end{align}
\end{MyClaim}
\begin{proof}[Proof of \cref{clm:owp-2}]
Let 
\begin{align}
    T:=\{y:f_{\{x'\}}(y)\ne f_{\{x,x'\}}(y) \lor \eval^{f_{\{x'\}}}(y) \ne \eval^{f_{\{x,x'\}}}(y)\}.
\end{align}
Note that the (tuples of) oracles $S_1$ and $S_2$ only differ on inputs in $T$. First, we show that for all $q(\lambda)$-query adversaries $\cB$,
\begin{align}
    \Pr_{f,x,x'}\big[\cB^{W^{S_1}}(f(x'))\in T\big]\le \frac{2q(\lambda)}{2^\lambda}.
\end{align}
If $\cB$ could exceed this bound, then by setting $f'=f_{\{x'\}}$ and $ f''= f_{\{x,x'\}}$ in \cref{clm:owp-find}, running $\mathsf{Find}^{f_{\{x'\}}}$ on the output of $\cB$ would result in $x$ with probability greater than $1/2^\lambda$. Since a call to $\mathsf{Find}^{f_{\{x'\}}}$ on an input of size at most $q(\lambda)$ can be implemented with at most $q(\lambda)$ calls to $W^{S_1}$, this gives a $2q(\lambda)$ query algorithm that finds $x$ with probability greater than $1/2^\lambda$, contradicting Claim \ref{clm:owp-1}. 

Let $\cM$ be the measurement that measures in the computational basis and accepts if the output is $x$. By the OW2H-compatibility of $W$ for $\bbS$, applying Theorem \ref{thm:ow2h-punc-full} to measurement $\cM$ and (tuples of) oracles $S_1$ and $S_2$, for all $q(\lambda)$-query adversaries $\cB$, \takashi{It may be useful to mention that we rely on OW2H-compatibility of $W$ here. A similar comment applies whenever we use Theorem \ref{thm:ow2h-punc-full}.}\kabir{better?}\takashi{looks good}
\begin{align}
\Big|\Pr_{f,x,x'}\big[\cB^{W^{S_1}}(f(x'))=x\big]-\Pr_{f,x,x'}\big[\cB^{W^{S_2}}(f(x'))=x\big]\Big|\leq \frac{\poly(q(\lambda))}{2^{c\lambda}},
\end{align}
which along with \cref{clm:owp-1} 
implies the claim.
\end{proof}

Next, we show that, when using the oracle $S_2$, the symmetry between $x$ and $x'$ implies that the probability that $\cB$ outputs $x$ is identical, regardless of whether its input is $f(x)$ or $f(x')$, implying that it is hard to find $x$ given input $f(x)$ and query access to $W^{S_2}$. 
\begin{MyClaim}
\label{clm:owp-3}
For all functions $q:\N\to\N$, all $q(\lambda)$-query adversaries $\cB$, and all large enough $\lambda$,
\begin{align}
\Pr_{f,x,x'}\big[\cB^{W^{S_2}}(f(x))=x\big]=\Pr_{f,x,x'}\big[\cB^{W^{S_2}}(f(x'))=x\big]\leq \frac{\poly(q(\lambda))}{{2^{c\lambda}}}.
\end{align}
\end{MyClaim}
\begin{proof}[Proof of \cref{clm:owp-3}]
Follows from Claim \ref{clm:owp-2} and the fact that $x$ and $x'$ are symmetric in the view of $\cB$.
\end{proof}

Next we show the upper bound of the probability that $\cB$ outputs $x$ with oracle $S_3$.
\begin{MyClaim}
\label{clm:owp-4}
    For all functions $q:\N\to\N$, all $q(\lambda)$-query adversaries $\cB$, and all large enough $\lambda$,
\begin{align}
\Prr_{f,x,x'}\big[\cB^{W^{S_3}}(f(x))=x\big]\leq \frac{\poly(q(\lambda))}{2^{c\lambda}}.
\end{align}
\end{MyClaim}

\begin{proof}[Proof of \cref{clm:owp-4}]
First, note that for all adversaries $\cB$, 
\begin{align}
\label{eq:owp-5}
    \Pr_{f,x,x'}\left[\cB^{W^{S_3}}(f(x))=x'\right]\le \frac{1}{2^\secp}
\end{align}
because $x'$ is sampled independent of the $\cB$’s view.
Let 
\begin{align}
    T:=\{y:f_{\{x,x'\}}(y)\ne f_{\{x\}}(y) \lor \eval^{f_{\{x,x'\}}
    }(y)\ne \eval^{f_{\{x\}}
    }(y) \}. 
\end{align}
\kabir{we don't need to number T because it is only defined in the scope of the claim}
We show that for all $q(\lambda)$-query adversaries $\cB$,
\begin{align}
    \Pr_{f,x,x'}\big[\cB^{W^{S_3}}(f(x))\in T\big]\le \frac{2q}{2^\lambda}.
\end{align}
If $\cB$ could exceed this bound, then by setting $f'=f_{\{x\}}$ and $ f''= f_{\{x,x'\}}$ in \cref{clm:owp-find}, running $\mathsf{Find}^{f_{\{x\}}}$ on the output of $\cB$ would result in $x$ with probability greater than $1/2^\lambda$. Since a call to $\mathsf{Find}^{f_{\{x\}}}$ on an input of size at most $q(\lambda)$ can be implemented with at most $q(\lambda)$ calls to $W^{S_3}$, this gives a $2q(\lambda)$ query algorithm that finds $x'$ with probability greater than $1/2^\lambda$, contradicting \cref{eq:owp-5}. 

Let $\cM$ be the measurement that measures in computational basis and accepts if the output is $x$. By the OW2H-compatibility of $W$ for $\bbS$, applying Theorem \ref{thm:ow2h-punc-full} to measurement $\cM$ and (tuples of) oracles $S_2$ and $S_3$, for all $q(\lambda)$-query adversaries $\cB$, 
\begin{align}
    \left|\Pr_{f,x,x'}\left[\cB^{W^{S_2}}(f(x))=x\right]-\Pr_{f,x,x'}\left[\cB^{W^{S_3}}(f(x))=x\right]\right|\le \frac{\poly(q(\lambda))}{2^{c\lambda}},
\end{align}
which along with \cref{clm:owp-3} implies the claim.
\end{proof}

Finally, we bound the probability that $\cB$ outputs $x$ when the oracle is $S_4$.
\begin{MyClaim}
\label{clm:owp-6}
For all $q:\N\to\N$ such that $q(\secp)\le2^{o(\secp)}$, all $q(\lambda)$-query adversaries $\cB$, and all large enough $\lambda$,
\begin{align}
    \Pr_{f,x} \left[\cB^{W^{S_4}}(f(x))=x\right]\le \frac{\poly(q(\lambda))}{2^{c^2\lambda/2}}.
\end{align}
\end{MyClaim}
\begin{proof}[Proof of \cref{clm:owp-6}]
Let 
\begin{align}
    T :=\{y:f(y)\ne f_{\{x\}}(y) \lor \eval^{f}(y)\ne\eval^{f_{\{x\}}}(y)\}. 
\end{align}
We show that for all $q(\lambda)$-query adversaries $\cB$, for large enough $\lambda$
\begin{align}
    \Pr_{f,x,x'}\big[\cB^{W^{S_3}}(f(x))\in T\big]\le \frac{1}{2^{c\secp/2}}.
\end{align}
If $\cB$ could exceed the bound, then by setting $f'=f_{x}$ and $f''=f$ in \cref{clm:owp-find} running $\mathsf{Find}^{f_{\{x\}}}$ on the output of $\cB$ would result in $x$ with probability greater than $\frac{1}{2q(\lambda)\cdot2^{c\lambda/2}}$. Since a call to $\mathsf{Find}^{f_{\{x\}}}$ on an input of size at most $q(\lambda)$ can be implemented with at most $q(\lambda)$ calls to $W^{S_3}$, this gives a $2q(\lambda)$ query algorithm that finds $x$ with probability greater than $\frac{1}{2q(\lambda)\cdot2^{c\lambda/2}}$, which contradicts \cref{clm:owp-4} for large enough $\lambda$ since $q(\lambda) \leq 2^{o(\lambda)}$. 

Let $\cM$ be the measurement that measures in computational basis and accepts if the output is  $x$. By the OW2H-compatibility of $W$ for $\bbS$, applying Theorem \ref{thm:ow2h-punc-full} to measurement $\cM$ and (tuples of) oracles $S_3$ and $S_4$, for all $q(\lambda)$-query adversaries $\cB$,
\begin{align}
\Big|\Prr_{S,x,x'}\big[\cB^{W^{S_3}}(f(x))=x\big]-\Prr_{S,x,x'}\big[\cB^{W^{S_4}}(f(x))=x\big]\Big|\leq \frac{\poly(q(\lambda))}{2^{c^2\lambda/2}},
\end{align}
which along with Claim \ref{clm:owp-4} implies the claim.
\end{proof}
\noindent Since $S_4=S$ and $q(\lambda)\le 2^{o(\lambda)}$, the theorem follows from Claim \ref{clm:owp-6}.
\end{proof}

\subsubsection{iO Secure Against Any OW2H-Compatible Unitary Collection}
We show that $S=(f,\obf,\eval^f)$ realizes iO for $f$-aided classical circuits  secure against any OW2H-compatible unitary collection. Here, the obfuscate and evaluate algorithms are allowed access to $S$ while the classical circuits to be obfuscated are allowed access to $f$. This captures obfuscation of circuits implementing a 
OWP, which is essential in most constructions of cryptographic primitives from iO. The adversary is allowed access to $S$ as well as the OW2H-compatible unitary.

To obfuscate an $f$-aided circuit $C$, the construction outputs $\obf(C,r)$ for randomly chosen $r$. To evaluate an obfuscated circuit $\obfC$ on input $x$, the construction returns $\eval^f(\obfC,r)$. The following theorem proves security.

\begin{theorem}
\label{thm:iO_security}
Let $S= (f, \obf, \eval^f)$ be the oracle defined in \cref{def:oracle_set} and let $\bbS$ be the set of all possible oracles with the same interface as $S$. Let $W=\{W^O\}_{O\in\bbS}$ be a collection of unitaries that are OW2H-compatible for $\bbS$ (\cref{def:punc}) with constant $c\in(0,1)$.
For all functions $q:\N \rightarrow \N$ such that $q(\lambda)\le 2^{o(\lambda)}$, all $q(\lambda)$-query adversaries $\cA$, and all large enough $\lambda$,
\begin{align}
\Big|\Prr_S\big[\cG^{\obf, f}_{0,\lambda}(\cA^{W^S})=1\big]-\Prr_S\big[\cG^{\obf, f}_{1,\lambda}(\cA^{W^S})=1\big]\Big|\leq \frac{1}{2^{c^2\lambda/2}},
\end{align}
where $\cG$ is the iO security game defined in \cref{def:io-game}.\kabir{there is an important subtlety here; this only rules out adversaries with positive advantage. However, since the game is efficient, standard techniques from \cite{bg11} show a conversion from standard adversaries to positive advantage adversaries; see also appendix B in \cite{TCC:BitDeg19} which has a proof custom to our setting} \kabir{actually standard io security only cares about positive advantage}
\end{theorem}

\begin{proof}
For any set $\bbX$ and all $x^*\in\{0,1\}^*$, we define
\begin{align}
    \obf_{\bbX}(x^*)=
    \begin{cases}
    \obf(x^*) & \text{if } x^*\notin\bbX,\\
    \bot & \text{if } x^*\in\bbX.
    \end{cases}
\end{align}
Also for any $r^* \in \bin^\lambda$ define the set 
\begin{align}
    (*,r^*) : = \{(C,r^*): C\in\bin^\lambda\}.
\end{align}
Next, for all $r^*\in\{0,1\}^\lambda$ define
\begin{align}
S_{r^*}:=\big(f,\obf_{(*,r^*)},\eval^{f}\big).
\end{align}
Note that $\eval^{f}$ continues to use $\obf^{-1}$ to compute $C$, and the change in oracles only affects oracle queries made by $C$. 
For all $b\in\{0,1\}$, for all $\lambda \in \N$, for any adversary $\cA$, define 
$\cG'_{b,\lambda,
\cA}$ \kabir{notation is weird, but I want to make it clear that here the oracle is chosen by the game as opposed to the io security game where the oracle given to the adversary is specified by running the game on $\cA^{W^S}$ }
as:
\begin{itemize}
\item $r\leftarrow\{0,1\}^\lambda$.
\item $(C_0,C_1)\leftarrow \cA^{W^{S_r}}(1^\lambda)$.
\item If $C^f_0$ is not functionally equivalent to $C^f_1$, or $C_1\notin\{0,1\}^{\lambda}$, or $C_0\notin\{0,1\}^{\lambda}$, return $\bot$.
\item Return $\cA^{W^{S_{r}}}\big(\obf(C_b,r)\big)$.
\end{itemize}

\begin{MyClaim}
\label{clm:obf-1}
For all adversaries $\cA$, \[
\Prr_S\big[\cG'_{0,\lambda, \cA}=1\big]=\Prr_S\big[\cG'_{1,\lambda,
\cA}=1\big].\]
\end{MyClaim}
\begin{proof}
In $\cA$’s view, $\obf(C_0,r)$ and $\obf(C_1,r)$ are sampled identically, uniformly at random under the constraint that they do not collide with $\obf(C,r')$ for any $r'\ne r$ or any $C\notin\{C_0,C_1\}$, and under the constraint that for all $x$
\begin{align}
\eval^{f}\big(\obf(C_0,r),x\big)=\eval^{f}\big(\obf(C_1,r),x\big)=C_0^{f}(x).
\end{align}
By symmetry both games are identical. \kabir{should explain this better}
\end{proof}
\noindent By Claim \ref{clm:obf-1} it suffices to show that for all $b\in\{0,1\}$, all functions $q:\N \rightarrow \N$ such that $q(\lambda)\leq 2^{o(\lambda)}$, all $q(\lambda)$-query adversaries $\cA$, and all large enough $\lambda$,
\begin{align}
\Big|\Prr_S\big[\cG^{\obf, f}_{b,\lambda}(\cA^{W^S})=1\big]-\Prr_S\big[\cG'_{b,\lambda,\cA}=1\big]\Big|\le \frac{1}{2\cdot 2^{c^2\lambda/2}}.
\end{align}
To prove this, define a sequence of random variables as follows. Let $r'\leftarrow\{0,1\}^\lambda$ and $r\leftarrow\{0,1\}^\lambda$. Define
\begin{align}
S_1&:=(f,\obf_{(*,r')},\eval^{f})\\
S_2&:=(f,\obf_{(*,r') \cup (*,r)},\eval^{f})\\
S_3&:=(f,\obf_{(*,r)},\eval^{f})\\
S_4&:=(f,\obf,\eval^{f}).
\end{align}
Note that the $S_i$s are correlated with the choice of $r$ and $r'$ which are both sampled uniformly at random as stated above.
For any $r^*\in\bin^\lambda$ let $\obf(*,r^*):= \{\obf(C,r^*): C \in \bin^\lambda\}$. Note that the random variables are correlated with $r$ and $r'$ which are both sampled uniformly at random as specified above.

\begin{MyClaim}
\label{clm:obf-2}
For all adversaries $\cB$,
\begin{align}
\Prr_{r,r',S}\big[\cB^{W^{S_1}}(\obf(*,r'))=r\big]\le \frac{1}{2^\lambda}.
\end{align}
\end{MyClaim}
\begin{proof}
Follows from the observation that $r$ is sampled independent of the adversary’s view.
\end{proof}
\begin{MyClaim}
\label{clm:obf-3}
For all functions $q:\N \rightarrow \N$, all $q(\lambda)$-query adversaries $\cB$, and all large enough $\lambda$,
\begin{align}
\Prr_{r,r',S}\big[\cB^{W^{S_2}}(\obf(*,r'))=r\big]\le \frac{\poly(q(\lambda))}{2^{c\lambda}}.
\end{align}
\end{MyClaim}
\begin{proof}
Let $T:=\{y:\obf_{(*,r')}(y)\ne \obf_{(*,r') \cup (*,r)}(y)\}$. Note that the (tuples of) oracles in $S_1$ and $S_2$ can only differ on inputs in $T$.
Since if $y\in T$ then $y\in (*,r)$, and since \cref{clm:obf-2} bounds the probability of finding $r$,  for all functions $q:\N \rightarrow \N$ , all $q(\lambda)$-query adversaries $\cB$, and all large enough $\lambda$,
\begin{align}
\Prr_{r,r',S}\big[\cB^{W^{S_2}}(\obf(*,r'))\in T\big]\le \frac{1}{2^\lambda}.
\end{align} 
Let $\cM$ be the measurement that measures in computational basis and accepts if the output is  $r$. By the OW2H-compatibility of $W$ for $\bbS$, applying Theorem \ref{thm:ow2h-punc-full} to measurement $\cM$, (tuples of) oracles $S_1$ and $S_2$, and $z=(\obf(*,r'))$, for all $q(\lambda)$-query adversaries $\cB$, 
\begin{align}
\Big|\Prr_{r,r',S}\big[\cB^{W^{S_1}}(\obf(*,r'))=r\big]-\Prr_{r,r',S}\big[\cB^{W^{S_2}}(\obf(*,r'))=r\big]\Big|\le \frac{\poly(q(\lambda))}{2^{c\lambda}},
\end{align}
which along with Claim \ref{clm:obf-2} implies the claim.
\end{proof}

\begin{MyClaim}
\label{clm:obf-4}
For all functions $q:\N \rightarrow \N$, all $q(\lambda)$-query adversaries $\cB$, and all large enough $\lambda$,
\begin{align}
\Pr_{r,r',S}\Big[\cB^{W^{S_2}}(\obf(*,r))=r\Big]
= \Pr_{r,r',S}\Big[\cB^{W^{S_2}}(\obf(*,r'))=r\Big] \leq \frac{\poly(q(\lambda))}{2^{c\lambda}}.
\end{align}
\end{MyClaim}
\begin{proof}
Follows from Claim \ref{clm:obf-3} and the fact that $r$ and $r'$ are symmetric in the view of $\cB$.
\end{proof}

\begin{MyClaim}
\label{clm:obf-5}
For all adversaries $\cB$,
\begin{align}
\Pr_{r,r',S}\Big[\cB^{W^{S_3}}(\obf(*,r))=r'\Big] \leq \frac{1}{2^\lambda}.
\end{align}
\end{MyClaim}
\begin{proof}
Follows from the observation that $r'$ is sampled independently of the adversary’s view.
\end{proof}

\begin{MyClaim}
\label{clm:obf-6}
For all functions $q:\N \rightarrow \N$, all $q(\lambda)$-query adversaries $\cB$, and all large enough $\lambda$,
\begin{align}
\Pr_{r,r',S}\Big[\cB^{W^{S_3}}(\obf(*,r))=r\Big]\le \frac{\poly(q(\lambda))}{2^{c\lambda}}.
\end{align}
\end{MyClaim}
\begin{proof}
Let $T:=\{y:\obf_{(*,r)}(y)\ne \obf_{(*,r') \cup (*,r)}(y)\}$. Note that the (tuples of) oracles in $S_3$ and $S_2$ can only differ on inputs in $T$. Since if $y\in T$ then $y\in (*,r')$, and since \cref{clm:obf-5} bounds the probability of finding $r'$,  for all functions $q:\N \rightarrow \N$, all $q(\lambda)$-query adversaries $\cB$, and all large enough $\lambda$,
\begin{align}
\Pr_{r,r',S}\big[\cB^{W^{S_3}}(\obf(*,r))\in T\big]\le \frac{1}{2^\lambda}.
\end{align}
 Let $\cM$ be the measurement that measures in computational basis and accepts if the output is  $r$. By the OW2H-compatibility of $W$ for $\bbS$, applying Theorem \ref{thm:ow2h-punc-full} to measurement $\cM$, (tuples of) oracles $S_3$ and $S_2$, and $z=(\obf(*,r))$,  for all $q(\lambda)$-query adversaries $\cB$
\begin{align}
\Big|\Prr_{r,r',S}\big[\cB^{W^{S_2}}(\obf(*,r))=r\big]-\Prr_{r,r',S}\big[\cB^{W^{S_3}}(\obf(*,r))=r\big]\Big|\le \frac{\poly(q(\lambda))}{2^{c\lambda}},
\end{align}
which, along with Claim \ref{clm:obf-4}  implies the claim.
\end{proof}

\begin{MyClaim}
\label{clm:obf-7}
For all functions $q:\N \rightarrow \N$, all $q(\lambda)$-query adversaries $\cB$, and all large enough $\lambda$,
\begin{align}
\Big|\Prr_{r,r',S}\big[\cB^{W^{S_3}}(\obf(*,r))=1\big]-\Prr_{r,r',S}\big[\cB^{W^{S_4}}(\obf(*,r))=1\big]\Big|\le \frac{\poly(q(\lambda))}{2^{c^2\lambda}}.
\end{align}
\end{MyClaim}
\begin{proof}
Let $T:=\{y:\obf_{(*,r)}(y)\ne \obf(y)\}$. Note that the (tuples of) oracles in $S_3$ and $S_4$ can only differ on inputs in $T$. Since if $y\in T$ then $y\in (*,r)$, and since \cref{clm:obf-6} bounds the probability of finding $r$,  for all functions $q:\N \rightarrow \N$ , all $q(\lambda)$-query adversaries $\cB$, and all large enough $\lambda$,
\begin{align}
\Pr_{r,r',S}\big[\cB^{W^{S_3}}(\obf(*,r))\in T\big]\le \frac{\poly(q)}{2^{c\lambda}}.
\end{align}
Let $\cM$ be the measurement that measures in computational basis and accepts if the output is  $1$. By the OW2H-compatibility of $W$ for $\bbS$, applying Theorem \ref{thm:ow2h-punc-full} to measurement $\cM$, (tuples of) oracles $S_3$ and $S_4$, and $z=\obf(*,r)$, for all $q(\lambda)$-query adversaries $\cB$,
\begin{align}
\Big|\Prr_{r,r',S}\big[\cB^{W^{S_3}}(\obf(*,r))=1\big]-\Prr_{r,r',S}\big[\cB^{W^{S_4}}(\obf(*,r))=1\big]\Big|\le \frac{\poly(q(\lambda))}{2^{c^2\lambda}}
\end{align}
which concludes the proof of the claim.
\end{proof}
\noindent Now suppose there exists a function $q:\N \rightarrow \N$, a bit $b$, and a $q(\lambda)$-query adversary $\cA$, such that $q(\lambda)\leq 2^{o(\lambda)}$ and for all large enough $\lambda$,
\begin{align}
\Big|\Prr_S\big[\cG^{\obf,f}_{b,\lambda}(\cA^{W^S})=1\big]-\Prr_S\big[\cG'_{b,\lambda, \cA}=1\big]\Big|\ge \frac{1}{2\cdot 2^{c^2\lambda/2}}.
\end{align}
We use this to construct $\cB$ that contradicts Claim \ref{clm:obf-7}. For $k\in\{3,4\}$, let $\cB^{W^{S_k}}(\obf(*,r))$ perform the following:
\begin{itemize}
\item $(C_0,C_1)\leftarrow \cA^{W^{S_k}}(1^\lambda)$. 
\item If $C^f_0$ is not functionally equivalent to $C^f_1$, or $C_1\notin\{0,1\}^{\lambda}$, or $C_0\notin\{0,1\}^{\lambda}$, return $\bot$. 
\item Return $\cA^{W^{S_k}}\big(\obf(C_b,r)\big)$. Note that $\cB$ receives the entire set $\obf(*,r)$ and can therefore compute $\obf(C_0,r)$.
\end{itemize}
Since $S_3\equiv S_r$ and $S_4\equiv S$, this implies that 
\begin{align}
\Prr_{r,r',S}\big[\cB^{W^{S_3}}(\obf(*,r))=1\big] = \Prr_S\big[\cG'_{b,\lambda,
\cA}=1\big],
\end{align}
and
\begin{align}
\Prr_{r,r',S}\big[\cB^{W^{S_4}}(\obf(*,r))=1\big] = \Prr_S\big[\cG^{\obf,f}_{b,\lambda}(\cA)=1\big].
\end{align}
Therefore,
\begin{align}
\Big|\Prr_{r,r',S}\big[\cB^{W^{S_3}}(\obf(*,r))=1\big]-\Prr_{r,r',S}\big[\cB^{W^{S_4}}(\obf(*,r))=1\big]\Big|\ge \frac{1}{2\cdot 2^{c^2\lambda/2}} 
\end{align}
which, along with the observation that $q(\lambda) \leq 2^{o(\lambda)}$ so all for large enough $\lambda$, $\poly(q(\lambda))\leq 2^{o(\lambda)} < \frac{1}{2}\cdot2^{c^2\lambda/2}$,
contradicts Claim \ref{clm:obf-7}.
\end{proof}

\section{Separating Distributional Collision Resistance from Indistinguishability Obfuscation and One-Way Permutations}
\label{sec:CR}

In this section, we show the impossibility of fully black-box constructions of distributional collision resistant puzzles (dCRPuzzs)
(and therefore also distributional collision-resistant hashing (dCRH) and collision-resistant hash functions) from iO and OWPs.
We start with the definition of dCRH and dCRPuzzs, as well as the definition of a black-box construction of dCRPuzz from iO and OWPs.

\begin{definition}[Distributional Collision-Resistant Hashing (dCRH)~\cite{STOC:DubIsh06,EC:BHKY19}]
\label{def:dCRH}
Let $\{\cH_\secp : \bit^{n(\secp)} \to \bit^{m(\secp)}\}_{\secp\in\mathbb{N}}$
be an efficient function family ensemble. Here $n$ and $m$ are polynomials.
We say that it is a distributional collision-resistant hash (dCRH) function family 
if there exists a polynomial $p$ such that for any QPT algorithm $\cA$, 
\begin{align}
\SD(\{h, \cA(h)\}_{h\gets \cH_\secp}, \{h, \mathsf{Col}(h)\}_{h\gets \cH_\secp}) \ge\frac{1}{p(\secp)}
\end{align}
for all sufficiently large $\secp\in\mathbb{N}$.
Here $\mathsf{Col}(h)$ is the following distribution.
\begin{enumerate}
    \item 
    Sample $x\gets\bit^{n(\secp)}$.
    \item 
    Sample $x'\gets h^{-1}(h(x))$.
    \item 
    Output $(x,x')$.
\end{enumerate}
\end{definition}

\begin{definition}[dCRPuzz~\cite{MorShiYam_PDQP}]
\label{def:dCRPuzzs}
    A distributional collision-resistant puzzle (dCRPuzz) is a pair $(\Setup,\Samp)$ of algorithms such that
    \begin{itemize}
        \item $\Setup(1^\secp)\to\pp:$ A QPT algorithm that, on input the security parameter $\secp$, outputs a classical public parameter $\pp$.
        \item $\Samp(\pp)\to(\puzz,\ans):$ A QPT algorithm that, on input $\pp$, outputs two bit strings $(\puzz,\ans)$. 
    \end{itemize}
    We require the following property:
    there exists a polynomial $p$ such that for any QPT adversary $\cA$, for all sufficiently large $\secp\in\mathbb{N}$,
    \begin{align}
        \SD(\{\pp,\cA(\pp)\}_{\pp\gets\Setup(1^\secp)},\{\pp,\cC(\pp)\}_{\pp\gets\Setup(1^\secp)} )\ge\frac{1}{p(\secp)}   
    \end{align}
     where $\cC(\pp)$ is the following distribution:
    \begin{enumerate}
        \item Run $(\puzz,\ans)\gets\Samp(\pp)$.
        \item Sample $\ans'$ with the conditional probability $\Pr[\ans'|\puzz]\coloneqq\frac{\Pr[(\ans',\puzz)\gets\Samp(\pp)]}{\Pr[\puzz\gets\Samp(\pp)]}$. 
        \item Output $(\puzz,\ans,\ans')$.
    \end{enumerate}
    If $(\Setup,\Samp)$ has access to an oracle $O$ and satisfies the above property against any oracle-aided QPT adversary $\cA^O$, we say that $(\Setup,\Samp)$ is a dCRPuzz relative to an oracle $O$.
\end{definition}

\begin{definition}[Fully Black-Box Construction of dCRPuzz from Subexponentially-Secure iO and OWPs]
\label{def:BB_dCRPuzz}
    A fully black-box construction of a dCRPuzz from iO and OWPs consists of a pair of oracle-aided algorithms $(\Setup,\Samp)$, a polynomial $p$, a function $q:\N\rightarrow\N$ , and $q(\lambda)$ query algorithm $\cR$ such that  $q(\lambda) \leq 2^{o(\lambda)}$ and the following holds. 
    
    Let $f$ be any permutation (i.e. a length-preserving injective map on bitstrings), and let $(\obf,\eval)$ be any pair of functions that satisfy the correctness of iO for $f$-aided classical circuits, i.e., for all $\lambda\in\N$, for all oracle aided classical circuits $C \in \bin^\lambda$, for all $x \in \bin^*$ such that $|x|$ is the input length of $C$, and for all $r\in \bin^\lambda$
        \begin{align}
            \eval(\obf(C,r),x) = C^f(x).
        \end{align}
    Then the following properties must hold.
    \begin{itemize}
        \item \textbf{Correctness:} 
        $(\Setup^{f,\obf,\eval},\Samp^{f,\obf,\eval})$ satisfies the syntax of dCRPuzz, i.e. $\Setup^{f,\obf,\eval}(1^\lambda)$ outputs classical public parameter $\pp$ and $\Samp^{f,\obf,\eval}(\pp)$ outputs bitstrings $(\puzz, \ans)$.
        \item \textbf{Black-Box Security Proof:} 
       Let  $\cA$ be any oracle such that
        \begin{align}
        \label{eq:BB_break_dCRPuzz}
            \SD(\{\pp,\cA(\pp, r')\}_{\substack{\pp\gets\Setup^{f,\obf,\eval}(1^\secp)\\r' \leftarrow\bin^\lambda}},\{\pp,\cC(\pp)\}_{\pp\gets\Setup^{f,\obf,\eval}(1^\secp)} )<\frac{1}{p(\secp)}
        \end{align}
        \mor{$r'\gets\bit^\lambda$ is typo?}
        where $\cC(\pp)$ is the following distribution:
    \begin{enumerate}
        \item Run $(\puzz,\ans)\gets\Samp^{f,\obf,\eval}(\pp)$.
        \item Sample $\ans'$ with the conditional probability $\Pr[\ans'|\puzz]\coloneqq\frac{\Pr[(\ans',\puzz)\gets\Samp^{f,\obf,\eval}(\pp)]}{\Pr[\puzz\gets\Samp^{f,\obf,\eval}(\pp)]}$. 
        \item Output $(\puzz,\ans,\ans')$.
    \end{enumerate}
        Then there exists a polynomial $p'$ such that for infinitely many $\secp$, 
        \begin{align}
        \label{eq:BB_break_OWP}
            \Pr_{x\leftarrow\bin^\lambda} \left[ \cR^{\cA,f,\obf,\eval} (f(x))=x \right] \ge \frac{1}{p'(\secp)}
        \end{align}
        or 
        \begin{align}
        \label{eq:BB_break_iO}
            \left| \Pr\left[ \cG^{\obf, f}_{0,\lambda}(\cR^{\cA,f,\obf,\eval})=1\right]-\Pr\left[\cG^{\obf, f}_{1,\lambda}(\cR^{\cA,f,\obf,\eval})=1 \right] \right| \ge \frac{1}{p'(\secp)}
        \end{align}
         where  $\cG$ is the iO security game defined in \cref{def:io-game}.
    \end{itemize}
\end{definition}
To rule out black-box constructions, we will define two oracles $S$ and $\mathsf{Col}^S$. $S$ is the oracle defined in \cref{def:oracle_set}, while $\mathsf{Col}^S$ is an oracle that will allow us to break every construction of collision resistant primitives using $S$, while preserving the security of iO and OWPs constructed from $S$ in \cref{subsec:io-owp-from-s}. First we define $\mathsf{Col}^O$ for an arbitrary oracle $O$.
\begin{definition}[Oracle $\mathsf{Col}$]
\label{def:Col} 
    For any oracle $O$, and string $C$, let $\ket{\psi^O_C}$ be defined as follows:
\begin{enumerate}
    \item If $C$ is not a valid encoding of an oracle-aided quantum circuit that implements some unitary with two output registers $\reg{A}$ and $\reg{B}$, set $\ket{\psi^O_C} := \ket{\bot}.$
    \item Let $C^O \ket{0} = \sum_s \sqrt{p_s}\ket{s}_{\reg{A}}\ket{\psi_s}_{\reg{B}}$, where $\forall s, p_s \geq 0$ and $\ket{\psi_s}$ is some pure state. 
    \item $\ket{\psi^O_C} := \sum_s \sqrt{p_s}\ket{s}\ket{\psi_s}\ket{\psi_s}$.
\end{enumerate}
Let $\DCol^O_C$ be defined as the distribution obtained by measuring $\ket{\psi^O_C}$ in the computational basis, and let $\DCol^O := \bigotimes_C \DCol^O_C$ be a product distribution in the manner described in \cref{sec:def_comp}. The oracle $\CCol^O$ is sampled from $\DCol^O$. 
\end{definition}

Recall that we showed in \cref{subsec:io-owp-from-s} that $S$ can be used to instantiate iO and OWPs that remain secure in the presence of any OW2H-compatible unitary oracle. Recall also that by \cref{thm:comp-oracle}, query access to $\CCol^S$ can be simulated using access to the compression unitary for $\DCol^S$. In the sequel we show that the compression unitary is OW2H-compatible, which allows us to use the results from \cref{subsec:io-owp-from-s} to show that the instantiations remain secure in the presence of $\CCol^S$.


\begin{theorem}
    \label{thm:CCol-punc}
    For any oracle $O$, let $W^O$ be the compression unitary (\cref{def:comp_U}) for $\DCol^O$. Then for any set of oracles $\bbO$, $W := \{ W^O \}_{O \in \bbO}$ is OW2H-compatible (\cref{def:punc}) for $\bbO$ with $c = 1/4$.
\end{theorem}
\begin{proof}[Proof of \cref{thm:CCol-punc}]
First note that queries to $O$ can be simulated by querying $W^O$ on $C$ that queries $O$, writes the output on the first output register, and leaves the second register empty, which satisfies the first requirement for OW2H-compatibility.
Let $\cB^{\cA, W^O}$ be the algorithm that samples $i \leftarrow [q]$, runs $\cA^{W^O}$, measures the query register of the $i$-th query in the computational basis, and obtains a measurement result $C$. $\cB^{\cA, W^O}$ then samples $j \leftarrow [q]$, runs $C^O$ (queries to $O$ can be simulated using $W^O$), measures the query register of the $j$-th query in the computational basis, and outputs the measurement result.
For any $O, O' \in \bbO$, for any $q$-query $\cA$,
\begin{align}
\epsilon := \left\| \ket{\psi} - \ket{\psi'} \right\|,
\end{align}
where $\ket{\psi}$ and $\ket{\psi'}$ are purified final states of $\cA^{W^O}$ and $\cA^{W^{O'}}$, respectively.
Then by Theorem \ref{thm:ow2h-comp},
\begin{align}
\mathbb{E}_C\Big[ \| \ket{\DCol^O_C} - \ket{\DCol^{O'}_C} \| \Big] \geq \frac{\epsilon^2}{32q^2}.
\label{eqColdifference}
\end{align}
where expectation is taken over the measurement of $C$ in $\cB$. Over the next few claims we will show that for any $C$ where $\| \ket{\DCol^O_C} - \ket{\DCol^{O'}_C}\|$ is noticeable, $C^O\ket{0}$ has noticeable query mass on points where $O$ and $O'$ differ. This will allow us to use \cref{lem:bbbv} to show that $\cB$ outputs a point where $O$ and $O'$ differ.
\begin{MyClaim}
\label{clm:real-distances-are-smaller}
    For any normalized states $\ket{\widetilde{\psi}} = \sum_x \alpha_x \ket{x}$ and $\ket{{\widetilde{\psi}'}} = \sum_x \alpha'_x \ket{x}$,
    we have
\begin{align}
\| \ket{\widetilde{\psi}} - \ket{{\widetilde{\psi}'}} \| \geq \| \ket{\hat{\psi}} - \ket{\hat{\psi'}} \|,
\end{align}
where $\ket{\hat{\psi}} := \sum_x |\alpha_x| \ket{x}$ and $\ket{\hat{\psi'}} := \sum_x |\alpha'_x| \ket{x}$.
\end{MyClaim}

\begin{proof}
\begin{align} 
\| \ket{\widetilde{\psi}} - \ket{{\widetilde{\psi}'}} \| &= \sqrt{2(1 - \mathsf{Re}\langle \widetilde{\psi} | {\widetilde{\psi}'} \rangle)}\\
\| \ket{\hat{\psi}} - \ket{\hat{\psi}'} \| &= \sqrt{2(1 - \mathsf{Re}\langle \hat{\psi} | \hat{\psi}' \rangle)}.
\end{align}
And since
\begin{align}
\mathsf{Re}\langle \widetilde{\psi}| {\widetilde{\psi}'} \rangle 
= \mathsf{Re}\left( \sum_x \alpha^*_x \alpha'_x \right)
\leq \left|\sum_x \alpha^*_x\alpha'_x\right|
\leq \sum_x |\alpha_x|\cdot|\alpha'_x|
=\mathsf{Re}\langle \hat{\psi} | \hat{\psi'} \rangle,
\end{align}
the statement of the claim follows.
\end{proof}
\noindent Recall that for any oracle $\widetilde{O}$
\begin{align}
    \ket{\DCol^{\widetilde{O}}_C} = \sum_z\sqrt{{\rm Pr}_{\DCol^{\widetilde{O}}_C}[z]} \ket{z}
\end{align}
and also recall that $\DCol^{\widetilde{O}}_C$ is the distribution obtained by measuring $\ket{\psi^{\widetilde{O}}_C}$ in the computational basis. 
Therefore, by \cref{clm:real-distances-are-smaller},
for all $C$,
\begin{align}
\| \ket{\DCol^O_C} - \ket{\DCol^{O'}_C} \| \leq \| \ket{\psi^O_C} - \ket{\psi^{O'}_C} \|.
\end{align}
With \cref{eqColdifference}, this implies
\begin{align}
\mathbb{E}_C\Big[ \| \ket{\psi^O_C} - \ket{\psi^{O'}_C} \| \Big] \geq \frac{\epsilon^2}{32q^2},
\end{align}
where expectation is over the $C$ sampled during the execution of $\cB$. The next claim relates $\| \ket{\psi^O_C} - \ket{\psi^{O'}_C} \|$ and the distance between $C^O \ket{0}$ and $C^{O'} \ket{0}$.

\begin{MyClaim}
\label{clm:abcd}
For all $C$
\begin{align}
\| \ket{\psi^O_C} - \ket{\psi^{O'}_C} \| \leq 3 \| C^O \ket{0} - C^{O'} \ket{0} \|.
\end{align}
\end{MyClaim}

\begin{proof}
Let $C^O \ket{0} = \sum_s \sqrt{p_s}\ket{s}\ket{\psi_s}$ 
and 
$C^{O'} \ket{0} = \sum_s \sqrt{p_s'}\ket{s}\ket{\psi_s'}$.
Let
\begin{align}
    \ket{\phi_0} &:= \sum_s \sqrt{p_s} \ket{s}\ket{\psi_s}\ket{\psi_s}, \\
    \ket{\phi_1} &:= \sum_s \sqrt{p'_s} \ket{s}\ket{\psi_s}\ket{\psi'_s}, \\
    \ket{\phi_2} &:= \sum_s \sqrt{p_s} \ket{s}\ket{\psi_s}\ket{\psi'_s}, \\
    \ket{\phi_3} &:= \sum_s \sqrt{p'_s} \ket{s}\ket{\psi'_s}\ket{\psi'_s}.
\end{align}
Define
\begin{align}
\Delta := \Big\| \sum_s \sqrt{p_s}\ket{s}\ket{\psi_s} - \sum_s \sqrt{p'_s}\ket{s}\ket{\psi'_s} \Big\|.
\end{align}
and
\begin{align}
\alpha := \left(\sum_s \sqrt{p_s}\bra{s}\bra{\psi_s}\right)\left(\sum_s \sqrt{p'_s}\ket{s}\ket{\psi'_s}\right) = \sum_s\sqrt{p_s p'_s}\langle\psi_s|\psi'_s\rangle.
\end{align}
Note that 
\begin{align}
\|\ket{\phi_0} - \ket{\phi_1}\| = \|\ket{\phi_2} - \ket{\phi_3}\|  = \Delta.
\end{align}
\dakshita{say why the above eqn is true? maybe say data processing inequality applied to the output of O or O'}
We now show a bound on $\| \ket{\phi_1} - \ket{\phi_2} \|$.
By expanding the definitions,
\begin{align}
    \mathsf{Re}\langle \phi_1 | \phi_2 \rangle &= \sum_s \sqrt{p_s p'_s}
    =\sum_s\left|\sqrt{p_s p'_s}\right|
    \geq\sum_s\left|\sqrt{p_s p'_s} \langle\psi_s|\psi'_s\rangle\right|
    \geq  \mathsf{Re}(\alpha).
\end{align}
So we may bound 
\begin{align}
\| \ket{\phi_1} - \ket{\phi_2} \| = \sqrt{2(1 - \mathsf{Re}\langle \phi_1 | \phi_2 \rangle)}
\leq  \sqrt{2(1 - \mathsf{Re}(\alpha))},
\end{align}
which by the observation that
\begin{align} 
\Delta = \left\| \sum_s \sqrt{p_s}\ket{s}\ket{\psi_s} - \sum_s \sqrt{p'_s}\ket{s}\ket{\psi'_s} \right\| = \sqrt{2(1 - \mathsf{Re}(\alpha)})
\end{align}
implies
\begin{align}
\|\ket{\phi_1} - \ket{\phi_2}\| \leq \Delta.
\end{align}
Putting all the bounds together and using the triangle inequality gives the statement of the claim.
\end{proof}

\noindent As a result of the above claim,
\begin{align}
\mathbb{E}_C\Big[ \| C^O \ket{0} - C^{O'} \ket{0} \| \Big] \geq \frac{1}{3} \cdot \frac{\epsilon^2}{32q^2}.
\end{align}
By Lemma \ref{lem:bbbv}
\begin{align}
    \Pr[\cB^{W^O}\in T] &\geq \Exp[\|C^O\ket{0} - C^{O'}\ket{0}\|^2]/q\\
    &\geq \left(\Exp[\|C^O\ket{0} - C^{O'}\ket{0}\|]\right)^2/q
    \geq  \frac{\epsilon^4}{96^2q^5},
\end{align}
which rearranges to give
\begin{align}
    \epsilon\leq 4\sqrt{6} q^{5/4}(\Prr[ \cB^{W^O} \in T])^{1/4}
\end{align}
concluding the proof of the theorem.
\end{proof}
We can now obtain the following corollary showing security of iO and OWP from $S$ in the presence of $\CCol^S$ 
\begin{corollary}
\label{cor:io-owp-with-col}
Let $S = (f,\obf,\eval^f)$ be the random oracle defined in \cref{def:oracle_set}, and let $\CCol^S$ be the random oracle defined in \cref{def:Col}.  
For all functions $q:\N\to\N$ such that $q(\secp)\leq2^{o(\lambda)}$, 
for all $q(\lambda)$-query adversaries $\cA$, 
and for all sufficiently large $\lambda$,
\begin{align}
\Prr_{S,\CCol^S, x\leftarrow\bin^{\lambda}}\Big[\cA^{S, \CCol^S}\big(f(x)\big)=x\Big]\leq\frac{1}{2^{c^2\lambda/2}}
\end{align}
and
\begin{align}
\Big|\Prr_{S, \CCol^S}\big[\cG^{\obf, f}_{0,\lambda}(\cA^{S, \CCol^S})=1\big]-\Prr_{S, \CCol^S}\big[\cG^{\obf, f}_{1,\lambda}(\cA^{S, \CCol^S})=1\big]\Big|\leq \frac{1}{2^{c^2\lambda/2}},
\end{align}
where $\cG$ is the iO security game defined in \cref{def:io-game}.
\end{corollary}
\begin{proof}
    Suppose there exists a function $q:\N\to\N$ 
and a $q(\lambda)$-query adversary $\cA$ such that $q(\secp)\leq 2^{o(\lambda)}$ and $\cA$ violates the bounds in the corollary. For any oracle $O$, let $W^O$ be the compression unitary (\cref{def:comp_U}) for $\DCol^O$. Then, by \cref{thm:comp-oracle}, the observation that the compressed oracle can be implemented by making two calls to the compression unitary for each call to the oracle, and the observation that queries to $S$ can be simulated by a single call to $W^S$, there exists a $2q(\lambda)$-query adversary $\cB$ such that for infinitely many $\lambda$, 
\begin{align}
\Prr_{S, x\leftarrow\bin^{\lambda}}\Big[\cB^{W^S}\big(f(x)\big)=x\Big]>\frac{1}{2^{c^2\lambda/2}}
\end{align}
or
\begin{align}
\Big|\Prr_{S}\big[\cG^{\obf, f}_{0,\lambda}(\cB^{W^S})=1\big]-\Prr_{S}\big[\cG^{\obf, f}_{1,\lambda}(\cB^{W^S})=1\big]\Big|> \frac{1}{2^{c^2\lambda/2}}.
\end{align}
However, since by \cref{thm:CCol-punc}, for any set of oracles $\bbO$ the collection of compression unitaries $\{W^O\}_{O\in\bbO}$ is OW2H-compatible for $\bbO$, this violates atleast one of the security guarantees of \cref{thm:owp} or \cref{thm:iO_security}.
\end{proof}

We have shown that there is a construction of iO and OWPs from $S$ that remains secure in the presence of $\CCol^S$. Next, we show that no construction of dCRPuzz from $S$ can remain secure in the presence of $\CCol^S$.

\begin{theorem}
    \label{thm:dCRPuzz_Col}
    Let $(\Setup,\Samp)$ be a pair of oracle-aided algorithms and let $O$ be a classical oracle such $(\Setup^O,\Samp^O)$ satisfy the syntax of dCRPuzz.
    Then, there exists an oracle-aided QPT algorithm $\cE$ such for all large enough $\secp$,
    \begin{align}
        \SD \left( \{\pp,\cC(\pp)\}_{\pp\gets\Setup^O(1^\secp)}, \{\pp,\cE^{\mathsf{Col}^O}(\pp)\}_{\pp\gets\Setup^O(1^\secp)} \right) \leq \negl(\lambda), 
    \end{align}
    where $\cC(\pp)$ is the following distribution:
    \begin{enumerate}
        \item Run $(\puzz,\ans)\gets\Samp^O(\pp)$.
        \item Sample $\ans'$ with the conditional probability $\Pr[\ans'|\puzz]\coloneqq\frac{\Pr[(\ans',\puzz)\gets\Samp^O(\pp)]}{\Pr[\puzz\gets\Samp^O(\pp)]}$. 
        \item Output $(\puzz,\ans,\ans')$.
    \end{enumerate}
    Here, $\mathsf{Col}^O$ is as defined in \cref{def:Col}.
\end{theorem}

\begin{proof}[Proof of \cref{thm:dCRPuzz_Col}]
Without loss of generality, we can assume that $\Samp^O(\pp)\to(\puzz,\ans)$ runs as follows.
\begin{enumerate}
    \item 
    Apply a unitary $V^O_\pp$ on $|0....0\rangle$ to generate a state, 
    \begin{align}
        V^O_\pp|0...0\rangle
        =\sum_{\puzz}\sqrt{p_\puzz} |\puzz\rangle_\regA |\psi_\puzz\rangle_{\regB\regC},
    \end{align}
    where $p_\puzz\ge 0$ for all $\puzz$.
    \item 
    Measure the register $\regA$ in computational basis to get $\puzz$.
    \item 
    Measure the register $\regB$ in computational basis to get $\ans$.
    \item 
    Output $(\puzz,\ans)$.
\end{enumerate}
We consider the oracle-aided QPT algorithm $\cE^{\mathsf{Col}^O}$ that behaves as follows:
\begin{enumerate}
    \item Take $\pp$ as input.
    \item Query $\mathsf{Col}^O$ on $V_\pp$ and obtain the result $(\puzz,\ans,\mathsf{junk},\ans',\mathsf{junk'})$.
    \item Output $(\puzz,\ans,\ans')$.
\end{enumerate}
Then, by the definition of $\mathsf{Col}^O$, the output distribution of $\cE^{\mathsf{Col}^O}$ is identical to $\cC(\pp)$\footnote{Technically, since the oracles are sampled in the beginning of the experiment and fixed, the algorithm as described is deterministic. This can be easily remedied by appending a sufficiently long random pad to the $\mathsf{Col}^S$ query that does not affect the functionality of the circuit. This recovers the claimed result upto negligible error in the statistical distance. \takashi{I believe this should be technically correct, but I'm wondering if it's okay to leave it informal. Perhaps, this problem disappears if we consider the unitary version of $\mathsf{Col}$?}}.
\end{proof}
Finally, we show the impossibility of fully black-box construction of dCRPuzz from iO and OWPs, which is the main result of this section.

\begin{theorem}
\label{thm:sep_dCRPuzz_iO}
    Fully black-box constructions of dCRPuzz from iO and OWPs do not exist.
\end{theorem}

\begin{proof}[Proof of \cref{thm:sep_dCRPuzz_iO}]Let $S = (f,\obf,\eval^f)$ be the random oracle defined in \cref{def:oracle_set}, and let $\CCol^S$ be the random oracle defined in \cref{def:Col}.  
 Let $(\Setup, \Samp)$ be a fully black-box construction of dCRPuzz from (subexponentially-secure) iO and one-way permutations for some $q(\lambda) \leq 2^{o(\lambda)}$ and $q(\lambda)$-query reduction $R$. Now, since $f$ is a permutation and $(\obf,\eval^{f})$ satisfy correctness for iO, by \cref{def:BB_dCRPuzz}, $(\Setup^S, \Samp^S)$ satisfy the correctness of dCRPuzz. This means that by \cref{thm:dCRPuzz_Col}, we have an oracle-aided QPT algorithm $\cE^{\mathsf{Col}^S}$ that satisfies \cref{eq:BB_break_dCRPuzz}.  By \cref{def:BB_dCRPuzz}, this means that there exists a polynomial $p'$ such that $\widetilde{R}^{S, \mathsf{Col}^S}:= R^{S, \cE^{\mathsf{Col}^S}}$ has advantage $1/p'(\lambda)$ in inverting $f$ on a random output or in winning the iO security game for $(\obf,\eval^{f})$, which contradicts the security properties shown in \cref{cor:io-owp-with-col}.
\end{proof}

There exist fully black-box constructions of public-key encryption, deniable encryption, oblivious transfer, and quantum money from iO and OWPs.
Moreover, there exists a fully black-box construction of dCRPuzz from collision-resistant hash functions or one-shot signatures%
~\cite{MorShiYam_PDQP}.
By combining these results, we have the following corollary:
\begin{corollary}
\label{cor:final-col}
    A fully black-box construction of 
    $Y\in\{$collision-resistant hash functions,  one-shot signatures$\}$
    from 
    $
    X\in\{$public-key encryption, deniable encryption, non-interactive ZK, trapdoor functions, oblivious transfer, quantum money$\}$ 
    does not exist.
\end{corollary}

\section{Separating Quantum Lightning from Indistinguishability Obfuscation and One-Way Permutations}
\label{sec:light}

In this section, we show the impossibility of fully black-box constructions of quantum lightning from iO and OWPs.
We start with the definition of quantum lightning and the fully black-box constructions.

\begin{definition}[Quantum Lightning~\cite{JC:Zhandry21}]
    A quantum lightning scheme is a pair of algorithms $(\Gen,\Ver)$ such that
    \begin{itemize}
        \item $\Gen(1^\secp)\to(s,\ket{\psi_s})$: A QPT algorithm that takes a security parameter $\secp\in\N$ as input and outputs a classical string $s$ and a quantum state $\ket{\psi_s}$.
        \item $\Ver(s,\sigma)\to\top/\bot$: A QPT algorithm that takes $s$ and a quantum state $\sigma$ as input and outputs $\top/\bot$.
    \end{itemize}
    We require the following correctness and security:
    \begin{itemize}
        \item \textbf{Correctness:} $\Pr[\top\gets\Ver(s,\ket{\psi_s}):(s,\ket{\psi_s})\gets\Gen(1^\secp)]\ge 1-\negl(\secp)$.
        \item \textbf{Security:} Consider the following game:
        \begin{enumerate}
            \item $\cA$ takes $1^\secp$ as input and outputs a classical string $s$ and a potentially entangled state $\sigma_{\regA\regB}$ over two registers $\regA$ and $\regB$.
            \item If both of $\Ver(s,\sigma_\regA)$ and $\Ver(s,\sigma_\regB)$ outputs $\top$, then $\cA$ wins. Otherwise, $\cA$ loses.
            Here, $\sigma_\regA$ and $\sigma_\regB$ are the reduced state of $\sigma_{\regA\regB}$ over the register $\regA$ and $\regB$, respectively.
        \end{enumerate}
        Then, for any QPT adversary $\cA$, $\Pr[\cA \text{ wins}]\le\negl(\secp)$.
    \end{itemize}
\end{definition}

\begin{definition}[Fully Black-Box Construction of Quantum Lightning from Subexponentially-Secure iO and OWPs]
\label{def:BB_QLight}
    A fully black-box construction of quantum lightning from iO and OWPs consists of a pair of oracle-aided algorithms $(\Samp,\Ver)$, a function $q:\N\rightarrow\N$, and a $q(\lambda)$ query algorithm $\cR$ such that  $q(\lambda) \leq 2^{o(\lambda)}$ and the following holds. 
    
    Let $f$ be any permutation (i.e. a length-preserving injective map on bitstrings), and let $(\obf,\eval)$ be any pair of functions that satisfy the correctness of iO for $f$-aided classical circuits, i.e., for all $\lambda\in\N$, for all oracle aided classical circuits $C \in \bin^\lambda$, for all $x \in \bin^*$ such that $|x|$ is the input length of $C$, and for all $r\in \bin^\lambda$
        \begin{align}
            \eval(\obf(C,r),x) = C^f(x).
        \end{align}
    Then the following properties must hold
    \begin{itemize}
        \item \textbf{Correctness:} 
        $(\Gen^{f,\obf,\eval},\Ver^{f,\obf,\eval})$ satisfies the syntax and the correctness of quantum lightning, i.e., $\Gen^{f,\obf,\eval}(1^\secp)$ outputs a classical string $s$ and a pure quantum state $\ket{\psi_s}$, $\Ver^{f,\obf,\eval}$ takes input bitstring $s$ and a quantum state $\sigma$ and outputs $\top/\bot$, and  $\Pr[\top\gets\Ver^{f,\obf,\eval}(s,\ket{\psi_s}):(s,\ket{\psi_s})\gets\Gen^{f,\obf,\eval}(1^\secp)]\ge 1-\negl(\secp)$.
        \item \textbf{Black-Box Security Proof:} For any polynomial $p$, and for any quantum algorithm $\cA$, if 
        \begin{align}
        \label{eq:BB_break_QLight}
            \Pr[\top\gets\Ver^{f,\obf,\eval}(s,\sigma_\regA) \land \top\gets\Ver^{f,\obf,\eval}(s,\sigma_\regB):(s,\sigma_{\regA\regB})\gets\cA(1^\secp)] \ge \frac{1}{p(\secp)}
        \end{align}
        holds for infinitely many $\secp$, then there exists a polynomial $p'$ such that
        \begin{align}
        \label{eq:BB_break_OWP_2}
            \Pr \left[ \cR^{\cA,f,\obf,\eval} (f(x))=x \right] \ge \frac{1}{p'(\secp)}
        \end{align}
        or 
        \begin{align}
        \label{eq:BB_break_iO_2}
            \left| \Pr\left[ \cG^{\obf,f}_{0,\lambda}(\cR^{\cA,f,\obf,\eval})=1\right]-\Pr\left[\cG^{\obf,f}_{1,\lambda}(\cR^{\cA,f,\obf,\eval})=1 \right] \right| \ge \frac{1}{r(\secp)}
        \end{align}
        holds for infinitely many $\secp$, where $\cG_{b,\secp}$ is the iO security game defined in \cref{def:oracle_iO}.
        Here, $\cR^\cA$ means that $\cR$ is allowed quantum query access to $A$,  $A^\dagger$, $A^*$, and $A^{\top}$ where $A$ is the unitary that implements the purified algorithm $\cA$.
    \end{itemize}
\end{definition}
To rule out black-box constructions, we will define two oracles $S$ and $\QCol^S$. $S$ is the oracle defined in \cref{def:oracle_set}, while $\QCol^S$ is a unitary oracle that will allow us to break every construction of quantum lightning using $S$, while preserving the security of iO and OWPs constructed from $S$ in \cref{subsec:io-owp-from-s}. First we define $\QCol^O$ for an arbitrary oracle $O$.

\begin{definition}[Unitary $\QCol$]
\label{def:QCol}
For any oracle $O$, and string $C$, let $\ket{\psi^O_C}$ be defined as follows:

\begin{enumerate}
    \item If $C$ is not a valid encoding of an oracle aided quantum circuit with two output registers $\reg{A}$ and $\reg{B}$, set $\ket{\psi^O_C} := \ket{\bot}.$
    \item Let $C^O \ket{0} = \sum_s \sqrt{p_s}\ket{s}_{\reg{A}}\ket{\psi_s}_{\reg{B}}$, where $\forall s, p_s \geq 0$ and $\ket{\psi_s}$ is some pure state.
    \item $\ket{\psi^O_C} := \sum_s \sqrt{p_s}\ket{s}\ket{\psi_s}\ket{\psi_s}$.
\end{enumerate}
Let $\QCol^O_C$ be the unitary defined as
\begin{align}
\QCol^O_C := \ket{\psi^O_C}\!\bra{\perp}+\ket{\perp}\!\bra{\psi^O_C} + (\bbI - \ket{\perp}\!\bra{\perp} - \ket{\psi^O_C}\!\bra{\psi^O_C}),
\end{align}
and let $\QCol^O$ be defined as
\begin{align}
\QCol^O\coloneqq
\sum_C \ket{C}\!\bra{C} \otimes \QCol^O_C.
\end{align}

\end{definition}

\begin{remark}
A note on conjugate, inverse, and transpose queries: 
We note that $\QCol^O$ is an involution (i.e. self-inverse) and that $(\QCol^O_C)^* = \QCol^O_{C^*}$, which means that a conjugate query can be efficiently implemented simply by conjugating the input circuit, therefore having query access to $\QCol^O$ is equivalent to having access to $(\QCol^O, (\QCol^O)^\dagger, (\QCol^O)^*, (\QCol^O)^\top)$ upto a small computational overhead.
\end{remark}

Recall that we showed in \cref{subsec:io-owp-from-s} that $S$ can be used to instantiate iO and OWPs that remain secure in the presence of any OW2H-compatible unitary oracle. We now show that for any set of oracles $\bbO$, the set $\{\QCol^O\}_{O\in\bbO}$ is OW2H-compatible for $\bbO$, which allows us to use the results from \cref{subsec:io-owp-from-s} to show that the instantiations remain secure in the presence of $\QCol^S$.

\begin{theorem}[OW2H-Compatibility of $\QCol$]
    \label{thm:QCol-punc}
    For any set of oracles $\bbO$, $\{ \QCol^O \}_{O \in \bbO}$ is OW2H-compatible for $\bbO$ with $c = 1/4$.
\end{theorem}
\begin{proof}[Proof of \cref{thm:QCol-punc}]
Identical to the proof of Theorem \ref{thm:CCol-punc} except Claim $\ref{clm:real-distances-are-smaller}$ is not needed since the invocation of Theorem \ref{thm:ow2h-comp} directly gives $
\mathbb{E}\Big[ \| \ket{\psi^O_C} - \ket{\psi^{O'}_C} \| \Big] \geq \frac{\epsilon^2}{32q^2}$.
\end{proof}

We can now obtain the following corollary showing security of iO and OWPs from $S$ in the presence of $\QCol^S$. 
\begin{corollary}
\label{cor:io-owp-with-qcol}
Let $S = (f,\obf,\eval^f)$ be the random oracle defined in \cref{def:oracle_set}, and let $\QCol^S$ be the unitary oracle defined in \cref{def:QCol}.  
For all functions $q:\N\to\N$ such that $q(\secp)\leq2^{o(\lambda)}$, 
for all $q(\lambda)$-query adversaries $\cA$, 
and for all sufficiently large $\lambda$,
\begin{align}
\Prr_{S, x\leftarrow\bin^{\lambda}}\Big[\cA^{S, \QCol^S}\big(f(x)\big)=x\Big]\leq\frac{1}{2^{c^2\lambda/2}}
\end{align}
and
\begin{align}
\Big|\Prr_{S}\big[\cG^{\obf, f}_{0,\lambda}(\cA^{S, \QCol^S})=1\big]-\Prr_{S}\big[\cG^{\obf, f}_{1,\lambda}(\cA^{S, \QCol^S})=1\big]\Big|\leq \frac{1}{2^{c^2\lambda/2}},
\end{align}
where $\cG$ is the iO security game defined in \cref{def:io-game}.
\end{corollary}
\begin{proof}
    By \cref{thm:CCol-punc}, for any set of oracles $\bbO$, the collection $\{\QCol^O\}_{O\in\bbO}$ is OW2H-compatible for $\bbO$. The corollary therefore follows from the observation that $S$ can be simulated using $\QCol^S$ and from the security guarantees of \cref{thm:owp} or \cref{thm:iO_security}.
\end{proof}

We have shown that there is a construction of iO and OWPs from $S$ that remains secure in the presence of $\QCol^S$. Next, we show that no construction of quantum lightning from $S$ can remain secure in the presence of $\QCol^S$.
\begin{theorem}
\label{thm:QLight_QCol}
    Let $(\Samp,\Ver)$ be a pair of oracle-aided algorithms and let $O$ be a classical oracle such that $(\Samp^O,\Ver^O)$ satisfy the syntax and the correctness of quantum lightning schemes.
    Then, there exists an oracle-aided QPT algorithm $\cE$ such that 
    \begin{align}
        \Pr[\top\gets\Ver^O(s,\sigma_\regA) \land \top\gets\Ver^O(s,\sigma_\regB):(s,\sigma_{\regA\regB})\gets\cE^{\QCol^O}(1^\secp)] \ge 1-\negl(\secp).
    \end{align}
\end{theorem}
\begin{proof}[Proof of \cref{thm:QLight_QCol}]
Without loss of generality, we can assume that $\Gen^O(\pp)\to(s,\ket{\psi_s})$ runs as follows.
\begin{enumerate}
    \item 
    Apply a unitary $U^O_\secp$ on $|0....0\rangle$ to generate a state, 
    \begin{align}
        U^O_\secp|0...0\rangle
        =\sum_{s}\sqrt{p_s} |s\rangle_\regA |\psi_s\rangle_{\regB},
    \end{align}
    where $p_s\ge 0$ for all $s$.
    \item 
    Measure the register $\regA$ in computational basis to get $s$.
    \item 
    Output $(s,\ket{\psi_s})$.
\end{enumerate}
We consider the oracle-aided QPT algorithm $\cE^{\QCol^O}$ that behaves as follows:
\begin{enumerate}
    \item Take $1^\secp$ as input.
    \item Query $\QCol^O$ on $\ket{U^O_\secp}\ket{\bot}$ and obtain 
    \begin{align}
        |U^O_\secp\rangle \otimes \sum_s \ket{s}_\regA\ket{\psi_s}_{\regB}\ket{\psi_s}_{\regB'}.
    \end{align}
    \item Measure $\regA$ in the computational basis and obtain the result $s$.
    \item Output $(s,\ket{\psi_s},\ket{\psi_s})$.
\end{enumerate}
Then, by the correctness of quantum lightning schemes, we have
\begin{align}
    \Pr&[\top\gets\Ver^O(s,\ket{\psi_s}) \land \top\gets\Ver^O(s,\ket{\psi_s}):(s,\ket{\psi_s},\ket{\psi_s})\gets\cE^{\QCol^O}(1^\secp)] \\ 
    &= \Pr[\top\gets\Ver^O(s,\ket{\psi_s}) \land \top\gets\Ver^O(s,\ket{\psi_s}):(s,\ket{\psi_s})\gets\Gen^O(1^\secp)] \\
    &\ge 1-\negl(\secp).
\end{align}
which concludes the proof.
\end{proof}
Finally
we show the impossibility of fully black-box construction of quantum lightning from iO and OWPs,
which is the main result of this section.

\begin{theorem}
\label{thm:sep_QLight_iO}
    Fully black-box constructions of quantum lightning from iO and OWPs do not exist.
\end{theorem}
\begin{proof}[Proof of \cref{thm:sep_QLight_iO}]
Let $S = (f,\obf,\eval^f)$ be the random oracle defined in \cref{def:oracle_set}, and let $\QCol^S$ be the unitary oracle defined in \cref{def:QCol}.  
 Let $(\Gen, \Ver)$ be a fully black-box construction of quantum lightning from (subexponentially-secure) iO and one-way permutations for some $q(\lambda) \leq 2^{o(\lambda)}$ and $q(\lambda)$-query reduction $R$. Now, since $f$ is a permutation and $(\obf,\eval^{f})$ satisfy correctness for iO, by \cref{def:BB_QLight}, $(\Setup^S, \Samp^S)$ satisfy the correctness of quantum lightning. This means that by \cref{thm:QLight_QCol}, we have an oracle-aided QPT algorithm $\cE^{\QCol^S}$ that satisfies \cref{eq:BB_break_QLight}.  By \cref{def:BB_QLight}, this means that there exists a polynomial $p'$ such that $\widetilde{R}^{S, \QCol^S}:= R^{S, \cE^{ \QCol^S}}$ has advantage $1/p'(\lambda)$ in inverting $f$ on a random output or in winning the iO security game for $(\obf,\eval^{f})$, which contradicts the security properties shown in \cref{cor:io-owp-with-qcol}.
\end{proof}
Similarly to \cref{cor:final-col}, we obtain
\begin{corollary}
\label{cor:final-qcol}
    A fully black-box construction of 
    quantum lightning
    from 
    $
    X\in\{$public-key encryption, deniable encryption, non-interactive ZK, trapdoor functions, oblivious transfer, quantum money$\}$ 
    does not exist.
\end{corollary}
\section{Separations using the Non-Collapsing Measurement Oracle}
\label{sec:PDQP}
In this section we show the non-collapsing measurement oracle $\cQ$, which solves problems in $\mathsf{SampPDQP}$,
does not break iO or\mor{and?} OWPs. 
Since $\mathsf{SampPDQP}$ can decide $\mathsf{SZK}$~\cite{ITCS:ABFL16}, and thus break homomorphic encryption, as well as non-interactive computational private
information retrieval~\cite{C:BogLee13,TCC:LiuVai16},
\label{sec:puncturingtheoremforPDQP} 
this result rules out the construction of either primitive in a fully black-box manner from iO and OWPs.

The argument proceeds almost identically to the separation in \cref{sec:CR}, with the oracle $\CCol^S$ replaced with $\cQ^S$. The only difference is that proving that the corresponding compression unitary is OW2H-compatible is significantly more involved. 

\subsection{Defining the Non-Collapsing Measurement Oracle Relative to Oracles}
We modify the definition of $\cQ$ to allow for circuits that make oracle queries.
\begin{definition}[Non-Collapsing Measurement Oracle with Oracle Queries] 
\label{def:q-with-oracle}Let $O$ be a classical oracle and
let $C\coloneqq (C_1^{(\cdot)}, M_1, C_2^{(\cdot)}, M_2, \ldots, C_{T}^{(\cdot)}, M_{T})$ be a classical description of a quantum circuit,
where $C_i$ is an $\ell$-qubit quantum circuit that may make a single quantum query to an oracle at the end of its execution, and $M_i$ is the computational-basis measurement on $m_i$ qubits with $m_i\le \ell$.
Let $\DQ^O_C$ be defined as the output distribution of the following process:
\begin{enumerate}
    \item Set $|\psi_0\rangle := |0\rangle$ and $v_0 := 0$. 
    \item
    For each $i \in [T]$:
\begin{itemize}
    \item Sample $|\psi_i\rangle \leftarrow M_i \cdot U_i^f|\psi_{i-1}\rangle$.\mor{$C_i^f$?}
    \item Sample $v_i$ by measuring $|\psi_i\rangle$ in the computational basis.
\end{itemize}
\item Output $(v_1,...,v_T)$
\end{enumerate}
Since measured registers can be copied, we assume without loss of generality that the measured registers of $\ket{\psi_i}$ are never modified after the measurement. 
Also without loss of generality, we assume that $C_i^f$ is an application of a unitary $C_i$ followed by an oracle call to $f$. 
We may additionally assume without loss of generality that the oracle is implemented as a phase query (since even oracles that output bitstrings can be implemented with phase oracles upto some polynomial factor blowup in queries).

Let let $\DQ^O := \bigotimes_C \DQ^O_C$ be a product distribution in the manner described in \cref{sec:def_comp}. The oracle $\cQ^O$ is sampled from $\DQ^O$.
\end{definition}
\subsection{OW2H-Compatibility Theorem}
Here we show the analogue of \cref{thm:CCol-punc} for $\cQ$, which is the only difference between the proof in \cref{sec:CR} and this section. The final theorem will therefore follow from identical arguments.

\begin{theorem}
\label{thm:PDQP-punc}
   For any oracle $O$, let $W^O$ be the compression unitary (\cref{def:comp_U}) for $\DQ^O$. Then for any set of oracles $\bbO$, $W := \{ W^O \}_{O \in \bbO}$ is OW2H-compatible (\cref{def:punc}) for $\bbO$ with $c = 1/4$.
\end{theorem}
\begin{proof}
First we can easily check that queries to $O$ can be simulated by querying $W^O$ on $C$, that queries $O$, writes the output on the first output register, and leaves the second register empty.

We show the second condition of OW2H-compatibility.
Let $\cB$ be the algorithm that samples $i \leftarrow [q]$, runs $\cA^{W^O}$, measures the query register of the $i$-th query in the computational basis, and obtains measurement $C$. $\cB$ then samples $j \leftarrow [q]$, runs $C^O$ (queries to $O$ can be simulated using $W^O$), measures the query register of the $j$-th query in the computational basis, and outputs the measurement result.
For any $O, O' \in \mathcal{O}$, for any $q$-query $\cA$,
define
$\epsilon := \left\| \ket{\psi} - \ket{\psi'} \right\|$.
Then by Theorem \ref{thm:ow2h-dist},
\begin{align}
\label{eq:pdqp-bound}
\mathbb{E}\Big[ \SD\left(\DQ^O_C, \DQ^{O'}_C\right)\Big] \geq \frac{\epsilon^2}{16q^2}.
\end{align}
where the expectation is taken over $C$ measured during the execution of $\cB$.

Now we will show that any $C$ such that $\mathsf{SD}(\DQ_C^O, \DQ_C^{O'})$ is noticeable must be querying $T$ with a noticeable query mass, i.e., measuring a random query made by $C^O$ will result in $x \in T$ with noticeable probability. 


Fix any $C = (C_1, M_1, C_2, M_2, \ldots C_\tau, M_\tau)$ and define $\delta_{C}\coloneqq\mathsf{SD}(\DQ_C^O, \DQ_C^{O'})$. Let $v =(v_1 \ldots v_{\tau})$ be a random variable representing the output of $\DQ_C^O$ and let $|\psi_1\rangle \ldots |\psi_{\tau}\rangle$ be the corresponding states that are measured to sample the $v_i$s. Also let $v' = (v'_1 \ldots v'_{\tau})$ be a random variable representing the output of $\cQ_C^{O'}$. First note that each $v_i$ fixes the corresponding state $|\psi_i\rangle$. Since the distribution of $v_{i+1}$ is determined by $|\psi_i\rangle$, $v$ is a Markov process. Similarly, $v'$ is a Markov process. By Lemma $\ref{lem:markov-tv}$ 
\begin{align}
   \delta_{C} = \mathsf{SD}(v, v')
    \leq \sum_{i\in[\tau-1]} 2\cdot \mathsf{SD}(\{v_i, v_{i+1}\},\{v'_i, v'_{i+1}\}).
\end{align}
The value $v_i$ is generated by applying a sequence of unitaries and measurements to obtain a final state which is measured in the computational basis. We may instead consider deferring the measurements until right before sampling $v_i$. Suppose these measurements are applied to the first $k$ bits of the state. We therefore obtain 
\begin{align}
\ket{A} :=C^O_{i}C^O_{i-1}\ldots C^O_1 \ket{0}=  \sum_{s\in\bin^k} \alpha_s \ket{s}\ket{\phi_s}.
\end{align}
(where $\alpha_s \geq 0$). The first register is measured in the computational basis to obtain the state $\ket{s}\ket{\phi_s}$ which is then measured to obtain $v_i$. $v_{i+1}$ can then obtained by applying $C^O_{i+1}$ to $\ket{s}\ket{\phi_s}$ and measuring in the computational basis. The distribution obtained by sampling $v_i$ and $v_{i+1}$ may therefore be obtained by a computational basis measurement on the state 
\begin{align}
\ket{\widetilde{A}} := \sum_{s\in\bin^k} \alpha_s \ket{s}\ket{\phi_s} \otimes C^O_{i+1}\ket{s}\ket{\phi_s}.
\end{align} 
Likewise, $v_i'$ is generated by the same process except $O'$ is used instead of $O$. Therefore, if we define
\begin{align}
\ket{B}:= C^{O'}_{i}C^{O'}_{i-1}\ldots C^{O'}_1 \ket{0} =
\sum_{s\in\bin^k} \alpha'_s \ket{s}\ket{\phi'_s}
\end{align} 
(where $\alpha'_s \geq 0$), then the distribution obtained by sampling $v'_i$ and $v'_{i+1}$ may therefore be obtained by a computational basis measurement on the 
state 
\begin{align}
\sum_{s\in\bin^k} \alpha'_s \ket{s}\ket{\phi'_s} \otimes C^{O'}_{i+1}\ket{s}\ket{\phi'_s}. 
\end{align}
Finally, since we assume each $C_i$ is a unitary followed by a single query to the phase oracle, and computational basis measurement distributions are unaffected by phases, we may instead obtain $v'_i$ and $v'_{i+1}$ by a computational basis measurement on the state 
\begin{align}
\ket{\widetilde{B}}:=\sum_{s\in\bin^k} \alpha'_s \ket{s}\ket{\phi'_s} \otimes C^{O}_{i+1}\ket{s}\ket{\phi'_s}.
\end{align}
The term $\mathsf{SD}(\{v_i, v_{i+1}\},\{v'_i, v'_{i+1}\})$ is therefore upper bounded by  $\mathsf{TD}(\ket{\widetilde{A}},\ket{\widetilde{B}})$\footnote{We abuse notation to use $\mathsf{TD}(\ket{\sigma},\ket{\sigma'})$ to refer to $\mathsf{TD}(\ket{\sigma}\bra{\sigma},\ket{\sigma'}\bra{\sigma'})$ for any pure states $\ket{\sigma}$ and $\ket{\sigma'}$.}. 
\begin{MyClaim}
\begin{align}
        \mathsf{TD}(\ket{\widetilde{A}}, \ket{\widetilde{B}})\leq  3 \cdot \mathsf{TD}(\ket{{A}}, \ket{{B}}).
\end{align}
\end{MyClaim}
\begin{proof}
Let $\ket{\psi_s}:= C_{i+1}^O\ket{s}\ket{\phi_s}$ and $\ket{\psi'_s}:= C_{i+1}^O\ket{s}\ket{\phi'_s}$. To obtain a bound, we consider the following states.
\begin{align}
\ket{\widetilde{C}}&:=\sum_{s\in\bin^k} \alpha'_s \ket{s}\ket{\phi'_s} \ket{\psi_s}\\
\ket{\widetilde{D}}&:=\sum_{s\in\bin^k} \alpha_s \ket{s}\ket{\phi'_s} \ket{\psi_s}.
\end{align}
Now, by triangle inequality,
\begin{align}
    \label{eq:query-triangle-ineq}
    \mathsf{TD}(\ket{\widetilde{A}}, \ket{\widetilde{B}}) \leq 
    \mathsf{TD}(\ket{\widetilde{A}}, \ket{\widetilde{C}}) +
    \mathsf{TD}(\ket{\widetilde{C}}, \ket{\widetilde{D}}) + 
    \mathsf{TD}(\ket{\widetilde{D}}, \ket{\widetilde{B}}).
\end{align}
We will bound each of the terms on the RHS. First 
\begin{align}
    \mathsf{TD}(\ket{\widetilde{A}}, \ket{\widetilde{C}}) &= \sqrt{1 - |\langle\widetilde{A}|\widetilde{C}\rangle|^2}
     = \sqrt{1 - \left|\sum_s \alpha_s\alpha'_s\langle \phi_s|\phi'_s\rangle\langle \psi_s|\psi_s\rangle\right|^2}\\
    &= \sqrt{1 - \left|\sum_s \alpha_s\alpha'_s\langle \phi_s|\phi'_s\rangle\right|^2}
     = \sqrt{1 - \left|\langle A|B\rangle\right|^2}
     = \mathsf{TD}(\ket{A}, \ket{B}).
\end{align}
Similarly,
\begin{align}
    \mathsf{TD}(\ket{\widetilde{D}}, \ket{\widetilde{B}}) &= \sqrt{1 - |\langle\widetilde{D}|\widetilde{B}\rangle|^2}
     = \sqrt{1 - \left|\sum_s \alpha_s\alpha'_s\langle \phi'_s|\phi'_s\rangle\langle \psi_s|\psi'_s\rangle\right|^2}\\
    &= \sqrt{1 - \left|\sum_s \alpha_s\alpha'_s\langle \psi_s|\psi'_s\rangle\right|^2}
    = \sqrt{1 - \left|\sum_s \alpha_s\alpha'_s\langle \phi_s|\phi'_s\rangle\right|^2}\\
     &= \sqrt{1 - \left|\langle A|B\rangle\right|^2}
     = \mathsf{TD}(\ket{A}, \ket{B}).
\end{align}
Finally,
\begin{align}
    |\langle\widetilde{C}|\widetilde{D}\rangle| &= |\sum_s \alpha_s\alpha'_s|\\
    &= \sum_s |\alpha_s\alpha'_s|\\
    &\geq \sum_s |\alpha_s\alpha'_s \langle \phi'_s|\phi_s\rangle|\\
    &\geq |\sum_s \alpha_s\alpha'_s \langle \phi'_s|\phi_s\rangle| \\
    &\geq |\langle A|B\rangle|,
\end{align}
which implies that 
\begin{align}
    \mathsf{TD}(\ket{\widetilde{C}}, \ket{\widetilde{D}}) \leq \mathsf{TD}(\ket{A}, \ket{B}).
\end{align}
The claim follows by substituting in \eqref{eq:query-triangle-ineq}.
\end{proof}
The term $\mathsf{SD}(\{v_i, v_{i+1}\},\{v'_i, v'_{i+1}\})$ is therefore upper bounded as
\begin{align}
    \mathsf{SD}(\{v_i, v_{i+1}\},\{v'_i, v'_{i+1}\})&\leq 3\cdot \mathsf{TD}(\ket{A},\ket{B}) \\
    &\leq 6\cdot \|\ket{A} - \ket{B}\|\\
    &= 6\cdot \|(C^{O}_{i}C^{O}_{i-1}\ldots C^{O}_1  - C^{O'}_{i}C^{O'}_{i-1}\ldots C^{O'}_1 )\ket{0}\|.
\end{align} 
Let $\eta(\ket{\sigma})$ represent the probability mass of queries to $x \in T$ in $\ket{\sigma}$, and let $\eta_{C,i} := \sum_{j\in [i]}\eta(C^{O}_{j}C^{O}_{j-1}\ldots C^{O}_1 \ket{0})$, i.e. $\eta_{C,i}$ is the total query mass on $T$ in the $i$ queries made by $u$ to $O$. Then, by Lemma \ref{lem:bbbv}
\begin{align}
    \mathsf{SD}(\{v_i, v_{i+1}\},\{v'_i, v'_{i+1}\})&\leq 6 \cdot \sqrt{i\cdot \eta_{C,i}}.
\end{align} 
Summing over $i$
\begin{align}
    \sum_{i} \mathsf{SD}(\{v_i, v_{i+1}\},\{v'_i, v'_{i+1}\})&\leq 6 \cdot \sum_{i\in[q]}\sqrt{i\cdot \eta_{i}} \\
    &\leq 6q\sqrt{q\eta_{C,q}}.
\end{align}
 Next we use the bound on $\delta_{C}$
\begin{align}
   \delta_{C} \leq \sum_i 2\cdot \mathsf{SD}(\{v_i, v_{i+1}\},\{v'_i, v'_{i+1}\})
    \leq 12q^{3/2}\eta_{C,q}^{1/2},
\end{align}
which implies
\begin{align}
   \eta_{C,q} \geq 144\delta^2_{C}/q^3. 
\end{align}
Now conditioned on $\cB^{W^O}$ measuring $C$, the probability that $\cB^{W^O} \in T$ is at least $\eta_{C,q}/q\geq 144\delta^2_{C}/q^4$
Therefore, by Jensen's inequality, and the 
\begin{align}
    \Pr[\cB^{W^O}\in T] \geq 144\Exp[\delta^2_C]/q^4
    \geq 144\left(\Exp[\delta_C]\right)^2/q^5
    \geq  144\frac{\epsilon^4}{16^2q^8},
\end{align}
which rearranges to give
\begin{align}
   \epsilon \leq \sqrt{4/3} q^{2}(\Prr[ \cB^{W^O} \in T])^{1/4},
\end{align}
concluding the proof of the theorem.
\end{proof}


By the same arguments as \ref{sec:CR}, replacing $\CCol$ with $\cQ$, we obtain the following corollaries. First we obtain that iO and OWP can exist even against adversries that can query $\cQ$.
\begin{corollary}
\label{cor:io-owp-with-q}
Let $S = (f,\obf,\eval^f)$ be the random oracle defined in \cref{def:oracle_set}, and let $\cQ^S$ be the random oracle defined in \cref{def:q-with-oracle}.  
For all functions $q:\N\to\N$ such that $q(\secp)\leq2^{o(\lambda)}$, 
for all $q(\lambda)$-query adversaries $\cA$, 
and for all sufficiently large $\lambda$,
\begin{align}
\Prr_{S,\cQ^S, x\leftarrow\bin^{\lambda}}\Big[\cA^{S, \cQ^S}\big(f(x)\big)=x\Big]\leq\frac{1}{2^{c^2\lambda/2}}
\end{align}
and
\begin{align}
\Big|\Prr_{S, \cQ^S}\big[\cG^{\obf, f}_{0,\lambda}(\cA^{S, \cQ^S})=1\big]-\Prr_{S, \cQ^S}\big[\cG^{\obf, f}_{1,\lambda}(\cA^{S, \cQ^S})=1\big]\Big|\leq \frac{1}{2^{c^2\lambda/2}},
\end{align}
where $\cG$ is the iO security game defined in \cref{def:io-game}.
\end{corollary}

Finally, we can obtain the following corollary.\mor{Isn't the following corollary already obtained from the result for $\mathsf{Col}^O$?}
\begin{corollary}
\label{cor:sep_PDQP_iO+OWP}
    There does not exist a fully black box construction of hard problems in SZK, homomorphic encryption, or non-interactive computational private
information retrieval from $
    X\in\{$public-key encryption, deniable encryption, non-interactive ZK, trapdoor functions, oblivious transfer, quantum money$\}$ 
\end{corollary}
\begin{proof}(Proof Sketch)
 The proof follows almost identically to the proof of \cref{thm:sep_dCRPuzz_iO}, with the additional observation that $\mathsf{SampPDQP}$ can decide $\mathsf{SZK}$~\cite{ITCS:ABFL16}, and thus break homomorphic encryption, as well as non-interactive computational private
information retrieval~\cite{C:BogLee13,TCC:LiuVai16}.
 \end{proof}

\ifnum\anonymous=1
\else
{\bf Acknowledgments.}
DK and KT were supported in part by AFOSR, NSF 2112890, NSF CNS-2247727 and a Google
Research Scholar award. 
This material is based upon work supported by the Air Force Office of
Scientific Research under award number FA9550-23-1-0543.
TM is supported by
JST CREST JPMJCR23I3,
JST Moonshot R\verb|&|D JPMJMS2061-5-1-1, 
JST FOREST, 
MEXT QLEAP, 
the Grant-in Aid for Transformative Research Areas (A) 21H05183,
and 
the Grant-in-Aid for Scientific Research (A) No.22H00522.
YS is supported by JST SPRING, Grant Number JPMJSP2110. A\c{C} was supported by the following grants of Vipul Goyal: NSF award 1916939, a gift from Ripple, a DoE NETL award, a JP Morgan Faculty Fellowship, a PNC center for financial services innovation award, and a Cylab seed funding award.

Part of this work was done while the authors were visiting the Simons Institute for the Theory of Computing.
\fi

\ifnum\submission=0
\bibliographystyle{alpha} 
\else
\bibliographystyle{splncs04}
\fi
\bibliography{abbrev3,crypto,reference,text}


\end{document}